\documentclass[12pt]{article}
\usepackage[numbers,sort]{natbib}
\usepackage{fullpage}
\usepackage{amsmath, amsthm, amssymb}
\usepackage{verbatim}
\usepackage[pagebackref=false,colorlinks=false]{hyperref}
\hypersetup{linkcolor=blue,filecolor=blue,citecolor=blue,urlcolor=blue}
\usepackage{cleveref}

\DeclareMathOperator*{\E}{{\mathrm{E}}}

\newtheorem{theorem}{Theorem}
\newtheorem{definition}{Definition}
\newtheorem{lemma}{Lemma}
\newtheorem{claim}{Claim}
\newtheorem{corollary}{Corollary}

\begin{document}

\title{Optimal Prize Design in Parallel Rank-order Contests}
%Optimal Prize Structure in Multiple Rank-order Contests}

% Single author syntax
\author{
Xiaotie Deng 
\\CFCS, School of CS\\Peking University\\Beijing, China\\
xiaotie@pku.edu.cn\and
Ningyuan Li \\CFCS, School of CS\\Peking University\\Beijing, China\\
liningyuan@pku.edu.cn\and
Weian Li \\School of Software\\Shandong University\\Jinan, China\\
weian.li@sdu.edu.cn\and
Qi Qi \\Gaoling School of Artificial Intelligence\\ Renmin University of China\\Beijing, China\\
qi.qi@ruc.edu.cn
}

% \institute{Center on Frontiers of Computing Studies, School of Computer Science, Peking University, Beijing, China\\
% \email{\{xiaotie, liningyuan, weian\_li\}@pku.edu.cn} \and
% Gaoling School of Artificial Intelligence, Renmin University of China, Beijing, China\\
% \email{qi.qi@ruc.edu.cn}
% }

\date{}
\maketitle

\begin{abstract}
    This paper investigates a two-stage game-theoretical model with multiple parallel rank-order contests. In this model, each contest designer sets up a contest and determines the prize structure within a fixed budget in the first stage. Contestants choose which contest to participate in and exert costly effort to compete against other participants in the second stage. First, we fully characterize the symmetric Bayesian Nash equilibrium in the subgame of contestants, accounting for both contest selection and effort exertion, under any given prize structures. Notably, we find that, regardless of whether contestants know the number of participants in their chosen contest, the equilibrium remains unchanged in expectation. Next, we analyze the designers' strategies under two types of objective functions based on effort and participation, respectively. For a broad range of effort-based objectives, we demonstrate that the winner-takes-all prize structure-optimal in the single-contest setting-remains a dominant strategy for all designers. For the participation objective, which maximizes the number of participants surpassing a skill threshold, we show that the optimal prize structure is always a simple contest. Furthermore, the equilibrium among designers is computationally tractable when they share a common threshold.
\end{abstract}

\section{Introduction}
\label{sec:intro}

Contest theory is an important field in economics, attracting considerable attention for its depiction of various competitive scenarios in the real world. Examples of contests range from athletic events like the Olympic games to professional competitions such as hackathons. Existing research in contest theory primarily focuses on the single contest setting, examining how to design contests to motivate participants and optimize various objectives such as the maximum or total effort of contestants.

However, the digital revolution and the rise of online platforms has led to an increase in the number of concurrent contests. For instance, platforms like TopCoder host thousands of software contests annually. Another example relevant to academics involves choosing among various conferences to submit their research papers. Each conference essentially acts as a contest deciding which submissions to accept, while authors evaluate each conference's acceptance rates and prestige to determine their best options. This shift towards multiple, simultaneous contests has garnered the attention of researchers. A recent survey in contest theory \cite{S20} emphasizes the challenges and significance of examining the economics of multiple contests.

Multi-contest environments are inherently more complex than the single contest settings. Firstly, contestants face decisions not only on which contests to enter, but also the allocation of efforts, which is further complicated by the uncertainty in other competing contestants. Secondly, contest designers also face competition from other designers and externalities that were less critical in single contest designs.
%which means that the optimal contest design in a single contest setting may not be optimal in a multi-contest environment. 

In this paper, we analyze contestants' strategic behaviors and the design of contests under multi-contest environments.
% To tackle with aforementioned problems, we generalize the classic contest model proposed by the seminal paper \cite{MS01} and build up the corresponding theoretical model in multi-contest setting, which consists of multiple contest designers. For each designer, he needs to set his own prize structure of rank-order contest, within the given budgets. After witnessing the configurations of all contests, each contestant decides which contest to participate in and how much effort will be executed in this contest. In our model, the contestant still aims to maximize the expected prize gained from selected contest, but we consider several designer's targets (not limited to the total effort of joined contestants). In fact, our model is a two-stage game: the first is among contest designers and the second is among contestants. 
To address the aforementioned challenges, we generalize the classical contest model introduced by \cite{MS01}, extending it to accommodate multiple simultaneous contests. Our model is structured as a two-stage game, involving the competition of designers and contestants respectively. Initially, each contest designer sets up a rank-order prize structure within a given budget. After observing the configurations of all contests, each contestant chooses a contest to participate in and decides her effort level. While contestants aim to maximize their prize minus the cost of effort, various objectives for contest designers beyond simply maximizing total effort are considered. 

\subsection{Main Contributions}

Our contributions and results are summarized as follows:
\begin{enumerate}
    \item In the subgame among contestants, given the prize structures of all contests, we comprehensively characterize the symmetric Bayesian Nash equilibrium (sBNE) of contestants, which includes both contest choice and effort level. To simplify the two-dimensional strategy, we first eliminate the randomness in effort, demonstrating that a contestant's effort deterministically depends on the chosen contest and her skill in an sBNE. Focusing on contest selection, we show that despite the costly and strategic efforts in our model, the choice of contests in equilibrium aligns with a simpler model without strategic efforts \cite{DLLQ22}. We also examine a variant model where the number of competitors is disclosed to contestants before deciding the efforts, finding that the sBNE is essentially unaffected, which provides insights for real-world contest design scenarios.
    \item For designers, we consider two types of objective functions, based on effort and participation respectively. The effort objective involves a non-negative linear combination of efforts, covering prevalent objectives studied in contest theory, such as the total or maximum effort, or the total effort of the top $k$ participants. We find that for a wide range of effort objectives, the winner-takes-all prize structure, which is effective in single-contest scenarios, constitutes a dominant strategy for all designers. 
    %\lny{remove? This finding implies that the optimal contest remains the same in both single-contest and multi-contest settings, thereby extending the results from \cite{DGLLL21} to the incomplete information setting.}
    \item The participation objective aims to maximize the number of participants whose skill reaches a threshold level. We show that for a designer with this objective, the optimal prize structure is a simple contest, which distributes equal amount of prizes to the first several contestants. Furthermore, when all designers share a common threshold, the equilibrium of designers can be computed efficiently.
\end{enumerate}

We put all proofs in appendix.
% 1. 

% 2. 

% 3. 

\subsection{Relate Works}
Our paper belongs to the interdisciplinary domain of economics and computer science, especially for multiple contest competition and rank-order contest design. Our paper extends the seminal model proposed by \cite{MS01}, which examines optimal rank-order contest design in a single-contest environment aimed at maximizing contestants' total effort. We generalize this classic model to a multi-contest setting, where contestants not only decide how much effort to exert but also choose which contest to enter. We also consider broader objective functions for designers, including the effort objective, which covers the well-studied objectives like total or maximum effort, and the participation objective, which aims to maximize the number of attracted contestants.  %Accordingly,  takes effort into consideration and investigates the optimal contest design in multiple rank-order contests.

\subsubsection{Multiple Contests}
The existing literature on multiple contests primarily focuses on analyzing the equilibrium strategies adopted by contestants within given contest mechanisms. Specifically, \cite{AM09} pioneer the investigation into two identical Tullock contests \cite{T08}, examining the contestants' equilibrium contest selection and prize structures of contests under different objectives. \cite{DV09} consider multiple auction-based crowdsourcing contests, deriving equilibrium outcomes in both symmetric and asymmetric settings. %\cite{B16} compares the performance of a multi-contest environment with a single contest, highlighting the benefits of the former when contestant abilities vary significantly. %\cite{AM18} delve into two distinct types of contests, high competitiveness and low competitiveness, utilizing similar prize structures. At the same period, \cite{MSV18} characterize contestants' behavior in the context of two contests featuring disparate prize structures, one offering high rewards with fierce competition and the other, low rewards with low competition. 
%\cite{HHKK18} model university enrollment as a contest, demonstrating that less capable students prefer to participate in multiple contests, whereas high-ability students opt for a single contest. 
\cite{JSY20} investigate two contests with different prize amounts, revealing that contests offering higher prizes attract more contestants, yet an individual contestant's effort remains unaffected by the number of participants. \cite{KK22} study a scenario where a contestant can participate in multiple contests, but the output in each contest is affected by an uncertainty variable, and demonstrate that increasing the number of contests that a contestant engages in can enhance the utility of the contest designer. \cite{DLLQ22} study a multiple rank-order contest model for simple contestants, whose strategy is only the choice of contest but not the effort level, with rankings directly determined by skills. \cite{EGG22} analyze the equilibrium of contestants in multiple equal sharing contests. \cite{DGLLQ24} study the competition among several pairwise lottery contests and characterize the equilibrium of contestants and designers. Recently, \cite{DGLLL21} compare the optimal contest design between single contest and multiple contest settings, under a complete information environment with contest structures based on a contest success function, and find that the optimal design for a single contest also proves optimal for multiple contests. Interestingly, our findings align with their result, showing that the winner-take-all prize structure, optimal in single contests, remains a dominant strategy in our multiple-contest setting. %This phenomenon, where the optimal strategy for a single contest setting also prevails in multiple contest settings, is intriguing, and our results further corroborate this observation.

\subsubsection{Single Contest}

A comprehensive survey on single-contest design is presented in \cite{S20}. Herein, our focus is primarily on works related to rank-order contests, categorizing them based on their objectives. Several studies have emphasized effort-based objectives, such as maximizing the total effort of contestants \cite{GH88,MS01}, the maximum effort \cite{KG03,CHS19}, and the total effort of the top $k$ contestants \cite{AS09}. Additionally, participation objectives in single-contest settings have been explored, including maximizing the number of participants \cite{GK16} and the number of participants exceeding a certain effort threshold \cite{EGG21}. \cite{GH12} aim to optimize a function that incorporates both the quantity of submissions and the quality of outputs. 

To the best of our knowledge, our work is the first to study contestants' equilibrium behavior encompassing contest selection and effort exertion, alongside the optimal contest design for various kinds of designer objectives, in the environment of multiple rank-order contests.

%Initially, \cite{GH88} establish the optimal rank-order contests for both identical and non-identical contestants, with respect to their abilities, when the objective is to maximize total outputs. \cite{MS01} extend this framework by considering the scenario where contestants' abilities are drawn from a publicly known distribution, designing optimal contests for various cost function types, including linear, concave, and convex. \cite{KG03} delve into procurement contests, examining both symmetric and asymmetric cases and deriving the corresponding optimal contest structures. \cite{AS09} focuses on crowdsourcing contests with a substantial number of participants, characterizing the asymptotically optimal prize structure to maximize the cumulative effort of top-ranked contestants. \cite{GH12} aim to maximize a class of functions encompassing the number of submissions and the quality of outputs.  \cite{GK16} consider a model where contestants' strategies are whether to participate in rather than effort levels. \cite{CHS19} propose an optimal crowdsourcing contest design leveraging auction theory principles. Furthermore, \cite{EGG21} investigate the threshold objectives of designers and derive the optimal contest structures. Compared with these papers, our paper focuses on the part of above designer's objectives and discuss the optimal contest design in multi-contest setting.

\section{Model and Preliminaries}
\label{sec:pre}

% \subsection{Rank-Order Contest}

In this section, we formally introduce our model. Consider a multi-contest setting, where there are $m$ rank-order contests and $n$ contestants. To avoid ambiguity, we let $j \in [m] = \{1,2,\cdots, m\}$ and $i \in [n] = \{1,2,\cdots, n\}$ denote a contest and a contestant, respectively. 

In our model, each contest designer $j\in[m]$ designs a rank-order prize structure $\vec{w}_j = (w_{j,1}, w_{j,2}, \cdots, w_{j,n})$, representing the amounts of $n$ prizes, where $w_{j,1}\geq w_{j,2} \geq \cdots\geq w_{j,n}\geq 0$. The total amount of prizes is constrained by designer $j$'s budget $T_j \geq 0$, i.e., $\sum_{k=1}^n w_{j,k}\leq T_j$. 

Each contestant $i\in[n]$ has a private skill $v_i>0$, which is identically and independently drawn from a public distribution $F$. For simplicity we assume $F$ is a continuous distribution. For convenience, we use quantiles to represent the skills. Specifically, define the quantile function $v(q):=F^{-1}(1-q)$, and note that each contestant $i$'s quantile $q_i$ is a random variable independently following $U[0,1]$. Then, we can represent $v_i$ as $v_i=v(q_i)$. In the following parts, we often use $q_i$ to represent the contestant $i$'s competitiveness. Note that $v(q)$ is strictly decreasing, and therefore a lower quantile indicates a higher skill.

Our model can be formalized as a two-stage game:
\begin{itemize}
    \item In the first stage, all contest designers design their prize structures, $\vec{w}_{1}, \vec{w}_{2}$ $, \cdots, \vec{w}_{m}$, and announce them to the contestants.
    \item In the second stage, based on prize structures of all contests, each contestant $i\in[n]$ chooses one contest to participate in, denoted by $J_i \in [m]$, and decides her effort $e_{i}$ to exert in the contest $J_i$, which incurs a cost $e_{i}/v_i$.  
\end{itemize}

After the decision stages, each contest designer $j$ allocates prizes by a rank-by-effort allocation rule. Let $I_j=\{i:J_i=j\}$ denote the participants in contest $j$. The contestants in $I_j$ are ranked by their efforts in descending order, and get the prizes in $\vec{w}_j$ according to their ranks\footnote{Ties are broken in decreasing order of skills. If the number of participants $|I_j|$ is less than $n$, only the first $|I_j|$ prizes are awarded.}: the contestant exerting the highest effort %(representing the highest skill) 
gets prize $w_{j,1}$, the contestant with the second highest effort gets prize $w_{j,2}$, and so forth.

Let $(\boldsymbol{J},\boldsymbol{e})$ denote a pure strategy profile of contestants, where $\boldsymbol{J}=(J_1,J_2,\cdots,J_n)$ and $\boldsymbol{e}=(e_1,e_2,\cdots,e_n)$ represent the choice of contests and the exerted efforts, respectively. Under the current pure strategy profile $(\boldsymbol{J},\boldsymbol{e})$, the amount of prize that a contestant $i$ receives from the contest $J_i$ is

$$W_i(\boldsymbol{J},\boldsymbol{e}):=w_{J_i,\mathrm{rank}(i,\boldsymbol{J},\boldsymbol{e})},$$
where $\mathrm{rank}(i,\boldsymbol{J},\boldsymbol{e})=|\{i'\in[n]\setminus\{i\}: (e_{i'}<e_i)\land (J_{i'}=J_i)\}|+1$ represents the rank of $i$ among all participants in contest $J_i$.

When contestant $i$ exerts an effort $e_i$, a linear cost $e_{i}/v_i$ is incurred. The utility of a contestant is the difference between her prize and her cost. Equivalently, we always scale the contestant $i$'s utility by $v_i$, and define her utility as

$$u_i(\boldsymbol{J},\boldsymbol{e};v_i):=v_i\cdot W_i(\boldsymbol{J},\boldsymbol{e})-e_i.$$

In this paper, we allow contestants to take mixed strategies. %on choice of contests. 
A mixed strategy of contestant $i$ is a joint distribution of $(J_i,e_i)$, denoted by $\tau\in \Delta([m]\times\mathbb{R}_{\geq0})$\footnote{The notation $\Delta(S)$ denotes the space of all distributions on a set $S$.}. Note that in the subsequent sections, we will show that randomization for efforts is unprofitable, that is, without loss of generality we can represent $\tau$ as a distribution on the contests with a deterministic effort in each possible contest.

Since each contestant's quantile is private, the subgame of contestants is of incomplete information. Throughout the paper, we assume that all contestants follow a symmetric mixed strategy $\tau(q):[0,1]\to\Delta([m]\times\mathbb{R}_{\geq0})$, which means that each contestant $i\in[n]$ draws her pure strategy as $(J_i,e_i)\sim\tau(q_i)$. Given this symmetry, the index $i$ will represent a generic, arbitrary contestant. We focus on the symmetric Bayesian Nash equilibrium in the subgame of contestants, defined as follows. %where the mixed strategy of a contestant only depends on her quantile.
\begin{definition}
For any contestant $i$ with quantile $q_i$, when all other contestants take the symmetric mixed strategy profile $\tau(q)$, and she takes pure strategy $(J_i,e_i)$, her expected utility is defined as
% \begin{align*}
%     \bar{u}(J_i,e_i;&\tau(q);q_i)\\
%     &:=\E_{\substack{\forall i'\neq i: q_{i'}\sim U[0,1],\\
%     (J_{i'},e_{i'})\sim \tau(q_{i'})}}[u_i(J_i,\boldsymbol{J}_{-i},e_i,\boldsymbol{e}_{-i};v(q_i))].
% \end{align*}
\begin{align*}
    \bar{u}(J_i,e_i;&\tau(q);q_i):=\E[u_i(J_i,\boldsymbol{J}_{-i},e_i,\boldsymbol{e}_{-i};v(q_i))].
\end{align*}
Here the expectation is taken for all other contestants $i'\neq i$, following $q_{i'}\sim U[0,1]$ and $(J_{i'},e_{i'})\sim \tau(q_{i'})$.

We say a symmetric mixed strategy $\tau(q)$ is a symmetric Bayesian Nash equilibrium (sBNE), if for any $i\in[n]$, $q_i\in[0,1]$, $J_i'\in[m]$, $e_i'\in\mathbb{R}_{\geq0}$, it holds that
\begin{align*}
    \E_{(J_{i},e_i)\sim \tau(q_{i})}[\bar{u}(J_i,e_i;\tau(q);q_i)]
    \geq \bar{u}(J_i',e_i';\tau(q);q_i).
\end{align*}
\end{definition}

Lastly, we define the utility of contest designers. In this paper, we mainly consider two kinds of designers' objectives, effort objective and participation objective.
\begin{definition}[Effort objective]\label{def:effort-objective}
    An effort objective is specified by a group of coefficients, $\vec{\alpha}_j=(\alpha_{j,1},\cdots,\alpha_{j,n})$, weighting the effort of contestants by rank. Under a symmetric mixed strategy $\tau(q)$, the utility of designer $j$ is defined as
    
    $$ R_j(\tau(q);\alpha_j)=\E\left[\sum_{k=1}^n\alpha_{j,k}e^{(k)}_j\right],$$
    where the expectation is taken for $q_i\sim U[0,1]$, $(J_i,e_i)\sim \tau(q_i)$ for all $i\in[n]$, and $e^{(k)}_j$ denotes the effort of the contestant ranked $k$ in contest $j$. More specifically, $e^{(k)}_j=e_i$ if there exists $i\in[n]$ that $J_i=j$ and $\mathrm{rank}(i,\boldsymbol{J},\boldsymbol{e})=k$, and otherwise $e^{(k)}_j=0$. 
\end{definition}
\begin{definition}[Participation objective]\label{def:part-objective}
    A participation objective is specified by a quantile threshold $\theta_j\in[0,1]$. Under a symmetric mixed strategy $\tau(q)$, the utility of designer $j$ is defined as
    
    $$R_j(\tau(q);\theta_j)=\sum_{i\in[n]}\Pr[J_i=j\land q_i\leq \theta_j],$$
    where $q_i\sim U[0,1]$, $(J_i,e_i)\sim \tau(q_i)$ for all $i\in[n]$. This is the expected number of participants in contest $j$ whose skills achieve $v(\theta_j)$, indicating that they are qualified enough.
\end{definition}

\section{Contestant Equilibrium}\label{sec:CE}
In this section, we analyze the symmetric Bayesian Nash equilibrium within the subgame of contestants, given the prize structures of all contests. First, in subsection \ref{subsec:relationship betweem effort and contest}, we demonstrate that the effort level is deterministic once a contest is chosen, which allows us to represent the equilibrium solely by the contest choices. Next, in subsection \ref{subsec:characterizarion of CE}, we provide a complete characterization of the contestant equilibrium. Surprisingly, although the contestants exert costly efforts strategically in our model, we find the contest choice in equilibrium consistent with a simpler model without efforts, previously studied in \cite{DLLQ22}. Finally, in subsection \ref{subsec: known number of players}, we present an interesting finding that, even when contestants know the number of competing participants in their chosen contests and then choose the effort level accordingly, the symmetric Bayesian Nash equilibrium remains unchanged in expectation.

%Interestingly, although each contestant faces a two-dimensional decision, the equilibrium behavior in choosing contests coincides with a simpler model proposed in \cite{DLLQ22}, where contestants are ranked by their skill rather than effort, and can only decide which contest to participate in.

\subsection{Representation of Equilibrium Strategy}\label{subsec:relationship betweem effort and contest}
% Different from the single-contest setting, contestants need to make a two-dimension strategy (choice of contest and exerted effort) in multi-contest setting, which is difficult to analyze the contestants' equilibrium strategy directly. On the other hand, since the prize structure is $n$ discrete amounts, it is also hard to exploit the common methods, like solving the first order condition to find the best response of each contestant. 
Different from the single-contest setting, the strategy of contestants in our model is two-dimensional, i.e., choosing a contest and determining the effort level. This complexity makes it challenging to analyze contestants' equilibrium strategies directly. Moreover, the rank-order prize structure complicates the analysis as the prize amount involves discrete ranks. To tackle these challenges, we simplify the representation of the utility and strategy of contestants.
% Additionally, with the prize structure comprising $n$ discrete amounts, conventional methods, such as solving the first-order conditions to identify each contestant's best response, become difficult to apply.

% Due to the above challenging, we leverage the concept of \emph{interim allocation function}, which can represent the expected prize obtained from one contest, to bridge the contestants' utilities and prize structures of contests. 
We start by representing a contestant's interim utility by the \emph{interim allocation function}, which encapsulates the prize structure of a contest. First, given any symmetric mixed strategy $\tau(q)$, we define

$$\Psi_j^{\tau(q)}(e):=\Pr_{\substack{q_i\sim U[0,1],\\(J_i,e_i)\sim\tau(q_i)}}[J_i=j\land e_i\geq e],$$
which represents the probability that a contestant selects contest $j$ and exerts an effort greater than $e$ under the strategy $\tau(q)$. 

Observe that for any contestant $i\in [n]$, if she takes pure strategy $(J_i,e_i)$ and all other contestants adopt $\tau(q)$, then the number of contestants ranked before $i$, that is, $\mathrm{rank}(i,\boldsymbol{J},\boldsymbol{e})-1$, follows the binomial distribution $B(n-1,\Psi_j^{\tau(q)}(e))$. Therefore, her expected prize amount is $\E[w_{J_i,\mathrm{rank}(i,\boldsymbol{J},\boldsymbol{e})}] =
\sum_{k=1}^nw_{J_i,k}\binom{n-1}{k-1}(\Psi_{J_i}^{\tau(q)}(e_i))^{k-1}(1-\Psi_{J_i}^{\tau(q)}(e_i))^{n-k}$.
% \begin{align*}
%     \E[&w_{J_i,\mathrm{rank}(i,\boldsymbol{J},\boldsymbol{e})}]=\\
%     &\sum_{k=1}^nw_{J_i,k}\binom{n-1}{k-1}\Psi_{J_i}^{\tau(q)}(e_i)^{k-1}(1-\Psi_{J_i}^{\tau(q)}(e_i))^{n-k}.
% \end{align*}
Viewing this as a function of $\Psi_{J_i}^{\tau(q)}(e_i)$, we define the interim allocation function induced by the prize structure $\vec{w}_j$ of contest $j$ as

$$x_{\vec{w}_j}(\phi):=\sum_{k=1}^nw_{j,k}\binom{n-1}{k-1}\phi^{k-1}(1-\phi)^{n-k}.$$
Then her expected utility taking $(J_i,e_i)$ can be written as
\begin{align*}
    \bar{u}(J_i,e_i;\tau(q);q_i)=v(q_i)x_{\vec{w}_{J_i}}(\Psi_{J_i}^{\tau(q)}(e_i))-e_i.
\end{align*}

%Similar to \cite{wine21}, 

Notably, when $w_{j,1}=w_{j,n}$, i.e., the $n$ prizes in contest $j$ are all equal, it holds for all $\phi\in[0,1]$ that $x_{\vec{w}_j}(\phi)=w_{j,1}$. On the other hand, when $w_{j,1}>w_{j,n}$, $x_{\vec{w}_j}(\phi)$ is strictly decreasing in $\phi$ on $[0,1]$. For convenience, we may use $x_j(\phi)$ to denote $x_{\vec{w}_j}(\phi)$ when there is no ambiguity. 
% Then, we can rewrite the utility function by the interim allocation function,
% $$\bar{u}(J_i,e_i;\tau(q);q_i)=v(q_i)x_{\vec{w}_{J_i}}(\Psi_{J_i}^{\tau(q)}(e_i))-e_i.$$ 

Typically, a mixed strategy $\tau$ can be interpreted as a marginal distribution of $J_i\in[m]$ and $m$ conditional distributions of $e_i$ for each possible $J_i$. However, we can demonstrate that randomization over $e_i$ is generally unnecessary in sBNE. Specifically, the effort level $e_i$ is determined once $J_i$ and $q_i$ are given. We formalize this in the following lemma.

% can be viewed as a marginal distribution of $J_i\in[m]$ with $m$ conditional distributions of $e_i$ conditioning on each possibility of $J_i$. However, we can show that it is generally unnecessary to consider randomization on $e_i$ in an sBNE. That is, in an sBNE, once that the probabilities across contests are decided, there exists a dominant effort level exerted in each contest. We present this conclusion in the following lemma. 

\begin{lemma}\label{lemma:SBNE-deterministic-effort}
Given the prize structures $\vec{w}_1,\cdots,\vec{w}_m$, suppose $\tau(q)$ is an SBNE. Then, for any contest $j\in[m]$ and any $q_i\in[0,1]$ such that $\Pr_{(J_i,e_i)\sim \tau(q_i)}[J_i=j]>0$, the distribution of $e_i$ in $\tau(q_i)$ conditioning on $J_i=j$ is deterministic, i.e., there exist some $\beta_j:[0,1]\to\mathbb{R}_{\geq0}$ such that $\Pr_{(J_i,e_i)\sim \tau(q_i)}[e_i=
\beta_j(q_i)|J_i=j]=1$.

Moreover, let $\vec{\pi}(q)=(\pi_1(q),\cdots,\pi_m(q))$ denote the marginal distribution of $J_i$ in $\tau(q)$, where $\pi_j(q_i):=\Pr_{(J_i,e_i)\sim \tau(q_i)}[J_i=j]$. 
Then $\beta_j$ can be represented as 

$$\beta_j(q_i)=\hat{\beta}_j(\Phi_{\pi_j(q)}(q_i)),$$ 
where $\hat{\beta}_j(\phi)$ is a strictly decreasing and continuous function, and $\Phi_{\pi_j(q)}(q_i):=\int_0^{q_i}\pi_j(q)dq$.
\end{lemma}

By \Cref{lemma:SBNE-deterministic-effort}, we can simplify the symmetric equilibrium strategy without loss of generality. Specifically, we can represent it as a randomized choice of contest with a deterministic effort level for each contest, denoted by $(\vec{\pi}(q),\vec{\beta}(q))$, where $\vec{\pi}(q)= (\pi_1(q),\cdots,\pi_m(q))$ is called the \textbf{choice strategy}, and $\vec{\beta}(q)=(\beta_1(q),\cdots,\beta_m(q))$ is called the \textbf{effort strategy}.
% By \Cref{lemma:SBNE-deterministic-effort}, without loss of generality, we can regard a symmetric equilibrium strategy as a randomized choice of contest with a deterministic effort level for each contest, denoted by $(\vec{\pi}(q),\vec{\beta}(q))$, where $\vec{\pi}(q)=(\pi_1(q),\cdots,\pi_m(q))$ is called the \textbf{choice strategy} and $\vec{\beta}(q)=(\beta_1(q),\cdots,\beta_m(q))$ is called the \textbf{effort strategy}.

Note that if $(\vec{\pi}(q),\vec{\beta}(q))$ constitutes an sBNE, then modifying $\vec{\pi}(q)$ on any zero-measure subset of $q$ will not violate the equilibrium condition, provided that the modified distributions remain supported on the best response set. This means that there can be multiple sBNEs that differ only on zero-measure sets. Moreover, since these sBNEs induce the same utility for all contestants and designers, it is unnecessary to distinguish between them. To mitigate this multiplicity, we adopt the approach in \cite{DLLQ22} and focus on the \textbf{cumulative choice strategy}, which serves as a succinct proxy for $\vec{\pi}(q)$.

% It is worth noting that, provided that $(\vec{\pi}(q),\vec{\beta}(q))$ is an sBNE, modifying $\vec{\pi}(q)$ on any zero-measure subset of $q$ will not violate the equilibrium condition, as long as the modified distributions are still supported on the best response set. This implies that there can be multiple sBNEs that differ only in zero-measure sets. Moreover, as they induce the same utility for all contestants and designers, it is unnecessary to distinguish between these sBNEs. To avoid this multiplicity, we adopt the approach in \cite{DLLQ22}, and focus on the cumulation of the choice strategy, which is a succinct proxy of $\vec{\pi}(q)$. %to aid us analyze the contestant equilibrium in next subsection.
\begin{definition}\label{def:cumulative behavior}
    A cumulative choice strategy is defined as $\vec{\Phi}(q)=(\Phi_1(q),\cdots,\Phi_m(q))$, which satisfies the following:
    \begin{enumerate}
        \item For any $j\in[m]$, $\Phi_j(q)$ is a continuous, non-negative and non-decreasing function on $[0,1]$.
        \item For any $j\in[m]$ and $q\in[0,1]$, $\sum_{j=1}^m \Phi_j(q)=q$.
    \end{enumerate}

    Given a choice strategy $\vec{\pi}(q)=(\pi_1(q),\cdots,\pi_m(q))$, define $\vec{\Phi}_{\vec{\pi}}(q)=(\Phi_{\pi_1}(q),\cdots,\Phi_{\pi_m}(q))$, 
    where 
    
    $$\Phi_{\pi_j}(q_i):=\int_0^{q_i}\pi_j(t)dt.$$
    We say $\Phi_{\pi_j}(q)$ is the cumulation of $\pi_j(q)$ and $\vec{\Phi}_{\vec{\pi}}(q)$ is the cumulation of $\vec{\pi}(q)$, .
    
    We call $\vec{\Phi}(q)$ a \textbf{cumulative equilibrium choice strategy} if $\vec{\Phi}(q)$ is the cumulation of a choice strategy $\vec{\pi}(q)$ in some sBNE.
\end{definition}
%By the result in \cite{DLLQ22},  we know that $\vec{\Phi}(q)$ is a cumulative choice strategy if and only if $\vec{\Phi}(q)$ is the cumulation of some choice strategy $\vec{\pi}(q)$, i.e., $\Phi_j(q)=\Phi_{\pi_j}(q)$ for all $j\in[m]$. Observe that $\Phi_j(q_i)=%\int_0^{q_i}\pi_j(t)dt= \Pr[J_{i'}=j\land q_{i'}<q_i]$, which is the probability that another contestant $i'$ participates $j$ while having better skill than $i$. 
By Definition \ref{def:cumulative behavior}, choice strategies that differ only on zero-measure sets will induce the same cumulative choice strategy through integration. Conversely, results in \cite{DLLQ22} indicate that for any cumulative choice strategy, a corresponding choice strategy can also be found. From now on, we represent the symmetric strategy by a cumulative choice strategy $\vec{\Phi}(q)$ and an effort strategy $\vec{\beta}(q)$. Under this representation, the contestant equilibrium is unique in most cases, which simplifies the characterization of the equilibrium.

%we know that given any choice strategy, we can construct the cumulative choice strategy by integration. Reversely, the result in \cite{DLLQ22} tells us that we also can find a corresponding choice strategy by any cumulative choice strategy. From now on, we represent the symmetric strategy by a cumulative choice strategy $\vec{\Phi}(q)$ and an effort strategy $\vec{\beta}(q)$. We will show that, under this representation, the contestant equilibrium is unique in most cases, which simplifies the characterization of the contestant equilibrium.

\subsection{Characterization of Contestant Equilibrium}\label{subsec:characterizarion of CE}

%Now we characterize the cumulative choice strategy and the effort strategy in contestant equilibrium. Firstly, we give a sufficient and necessary condition on sBNE, shown in Lemma \ref{lemma:contestant-equilibrium-effort-condition}, in which we build up the exact relationship between 

We now characterize the cumulative choice strategy and effort strategy in contestant equilibrium. First, we provide a necessary and sufficient condition which characterizes the effort strategy $\vec{\beta}(q)$ by the cumulative choice strategy $\vec{\Phi}(q)$ in sBNE, as detailed in Lemma \ref{lemma:contestant-equilibrium-effort-condition}.
\begin{lemma}\label{lemma:contestant-equilibrium-effort-condition}
Given the interim allocation functions $x_1(\phi),$ $\cdots,x_m(\phi)$, a cumulative choice strategy $\vec{\Phi}(q)$ and an effort strategy $\vec{\beta}(q)$ represent an sBNE, if and only if the following conditions are satisfied:
\begin{enumerate}
    \item For any $j\in[m]$ and all $q_i\in[0,1]$ with $\Phi_j'(q_i)>0$, the effort strategy is uniquely determined as
    
    $$\beta_j(q_i)=\int_{q_i}^1 v(t)\left(-x_j'(\Phi_j(t))\right)\Phi_j'(t)dt.$$
    \item It holds for all $q_i\in[0,1]$ that
    $$\{j\in[m]:\Phi_j'(q_i)>0\}\subseteq \arg\max_{j\in[m]} (v(q_i)x_j(\Phi_j(q_i))-\beta_j(q_i)).$$
\end{enumerate}
\end{lemma}

% \begin{proof}
%     \notqed
% \end{proof}

By condition 1 in \Cref{lemma:contestant-equilibrium-effort-condition}, the cumulative choice strategy $\vec{\Phi}(q)$ induces a unique effort strategy $\vec{\beta}(q)$, such that if $\vec{\Phi}(q)$ is a cummulative equilibrium choice strategy, then $(\vec{\Phi}(q),\vec{\beta}(q))$ form a sBNE. This reduces the two-dimension strategy to one-dimension. %with $\vec{\Phi}(q)$. 
Based on this, we can focus on the condition for $\vec{\Phi}(q)$ to form the equilibrium.

\begin{theorem}\label{thm:contestant-equilibrium-choice-condition}
Given the interim allocation functions $x_1(\phi),$ $\cdots,x_m(\phi)$, a cumulative choice strategy $\vec{\Phi}(q)$ is a cummulative equilibrium choice strategy if and only if for all $q_i\in[0,1]$,
% \begin{equation}
% \label{eq:argmax-xj(phi)}
%     \{j\in[m]:\Phi_j'(q_i)>0\}\subseteq \arg\max_{j\in[m]}x_j(\Phi_j(q_i)).
% \end{equation}
\begin{equation*}
    \{j\in[m]:\Phi_j'(q_i)>0\}\subseteq \arg\max_{j\in[m]}x_j(\Phi_j(q_i)).
\end{equation*}
\end{theorem}

\Cref{thm:contestant-equilibrium-choice-condition} refines the condition 2 in \Cref{lemma:contestant-equilibrium-effort-condition} into a more direct condition. Essentially, while all contestants aim to choose the contest offering maximum expected utility, their equilibrium behavior boils down to choosing the contest offering maximum expected prize. This mirrors a simpler model introduced in \cite{DLLQ22}, where contestants choose contests without exerting efforts, and are then ranked by skill levels. In that model, contestants with higher quantiles never affect those with lower quantiles, so intuitively, as $q$ moves from 0 to 1, a contestant with quantile $q$ evaluates the expected remaining highest prize $x_j(\Phi_j(q))$ in each contest $j$, and selects among the highest ones. This process coincides with our equilibrium condition in \Cref{thm:contestant-equilibrium-choice-condition}.

Building on the results in \cite{DLLQ22}, we present a detailed characterization of the cummulative equilibrium choice strategy in \Cref{theorem:contestant-equilibrium-wine22}. We record some notations:
Let $x_j^{-1}(x):=\max\{\phi:x_j(\phi)$ $\geq x\}$\footnote{Specifically, when $x_j(0)<x$, define $x_j^{-1}(x)=0$.}, $Q(x):=\sum_{j\in[m]}x_j^{-1}(x)$ and $Q^{-1}(q):=\max\{x:Q(x)\geq q\}$. Note that $x_j^{-1}(x), Q(x)$ and $Q^{-1}(q)$ are non-increasing functions.

\begin{corollary}\label{theorem:contestant-equilibrium-wine22}
When all $x_j(\phi)$ are strictly decreasing, the unique cumulative equilibrium choice strategy $\vec{\Phi}(q)$ is given by

$$\Phi_j(q)=x_j^{-1}(Q^{-1}(q)).$$
%where $x_j^{-1}(x):=\max\{h:x_j(h)\geq x\}$, $Q(x):=\sum_{j\in[m]}x_j^{-1}(x)$, $Q^{-1}(q):=\max\{x:Q(x)\geq q\}$. Specially, when $x_j(0)<x$, we define $x_j^{-1}(x)=0$.
%\lny{$x_j(\phi)$}

On the other hand, suppose there is a partition of contests $M_1\subseteq[m]$ and $M_2=[m]\setminus M_1$, such that for every $j\in M_1$, $x_j(\phi)$ is a constant on $[0,1]$, and that for every $j\in [m]\setminus M_2$, $x_j(\phi)$ is strictly decreasing. Denote $x^*=\max_{j\in M_1}x_j(0)$ as the maximum among the constant values in $M_1$, and let $M_1^*=\{j\in M_1:x_j(0)=x^*\}$ and $q^*=\min\{1,\sum_{j\in M_2}x_j^{-1}(x^*)\}$. Then all cumulative equilibrium choice strategies $\vec{\Phi}(q)$ are given by:
\begin{itemize}
    \item For all $j\in M_2$, it holds that $\Phi_j(q)=x_j^{-1}(Q^{-1}(q))$ for all $q\in[0,1]$.
    \item For all $j\in M_1^*$ and all $q\in[0,q^*]$, it holds that $\Phi_j(q)=0$. For all $q\in(q^*,1]$, the value of $\Phi_j(q)$  can be arbitrary as long as $\sum_{j\in M_1^*}\Phi_j(q)=q-q^*$.
    \item For all $j\in M_1\setminus M_1^*$ and all $q\in[0,1]$, it holds that $\Phi_j(q)=0$.
\end{itemize}
\end{corollary}
%\lny{explanation}

\Cref{theorem:contestant-equilibrium-wine22} shows that if all interim allocation functions are strictly decreasing (i.e., for all $j\in [m]$, $w_{j,1}>w_{j,n}$), then a unique contestant equilibrium exists. However, when there is some contest $j$ with $w_{j,1}=w_{j,n}$, the cummulative choice strategy can be non-unique. Despite this, each contestant's utility from any equilibrium within this set remains the same.
% Note that Theorem \ref{theorem:contestant-equilibrium-wine22} demonstrates that, when the prizes in every contest are not equal, i.e., for any $j\in [m]$, $w_{j,1}>w_{j,n}$ and the interim allocation function is strictly decreasing, there exists a unique contestant equilibrium. However, when there exists at least one contest $j$ taking the prize structure with $w_{j,1}=w_{j,n}$, we can state that all contestant equilibria lie on a convex set, but for each contestant, the utilities she gains from any contestant equilibrium are the same. 

%\lny{Note that in the second case, $Q(x)$ is discontinuous at $x^*$, }
\subsection{Model with Disclosed Number of Competitors}\label{subsec: known number of players}
%In the last section, we have fully characterized the contestant equilibrium, when all prize structures are fixed. 
In the model discussed so far, contestants decide the chosen contest and exerted effort $(J_i,e_i)$ simultaneously. In this subsection, we examine a variant of our model where each contestant first choose a contest and then decide the effort after knowing the number of competitors in the same contest, which may better reflect certain real-world scenarios. Specifically, in this variant, the second stage is divided into two substages. In the first substage, each contestant $i\in[n]$ only chooses the contest $J_i$ to participate in. In the second substage, after being informed of the number of participants in contest $J_i$ (denoted by $k_{J_i}:=|I_{J_i}|$), the contestant then decides her effort $e_i$. Note that the skills and quantiles are still private to each contestant.

% In our model discussed so far, the contestants decide the chosen contest and exerted effort $(J_i,e_i)$ simultaneously. In this subsection, we explore a variant of our model where each contestant can decide her effort after knowing the number of competitors in the same contest, which may be more suitable for some real scenarios. Specifically, in this variant model, the second stage is split into two substages. In the first substage, each contestant $i\in[n]$ only chooses the contest $J_i$ to participate in. In the second substage, she is informed of the number of participants in contest $J_i$, denoted by $k_{J_i}:=|I_{J_i}|$, and then decides the exerted effort $e_i$. Note that the skills and quantiles are still private to each contestant.

Compared to the original model, this variant offers greater flexibility by allowing contestants to adjust their efforts based on the number of competitors. However, our following analysis will demonstrate that the equilibrium behavior of contestants remains fundamentally unchanged. This provides the practical implication that disclosing the number of competitors in a contest may not significantly impact the outcome.

% Compared to the original model, this variant model offers greater flexibility by allowing contestants to consider the number of competitors when deciding their efforts. However, our following analysis will show that the equilibrium behavior of contestants remains fundamentally unchanged, implying that whether to disclose the number of competitors to contests is not crucial in contest design.

In this variant model, similar to the argument in Lemma \ref{lemma:SBNE-deterministic-effort}, randomization over $e_i$ is unnecessary once $q_i$, $J_i$ and $k_{J_i}$ are given. Thus, without loss of generality, we can represent a symmetric equilibrium strategy by a cumulative choice strategy $\vec{\Phi}(q)$, as defined before, and an effort strategy $\vec{\beta}(q)=(\beta_{j,k}(q))_{j\in[m],k\in[n]}$, where each contestant $i$ determines her effort as $e_i=\beta_{J_i,k_{J_i}}(q_i)$.
% In this variant model, by the similar argument as Lemma \ref{lemma:SBNE-deterministic-effort}, one can see that randomization on $e_i$ is also unnecessary once that $q_i$, $J_i$ and $k_{J_i}$ are given. Therefore, without loss of generality, we represent a symmetric equilibrium strategy of contestants by a cumulative choice strategy $\vec{\Phi}(q)$, as defined before, and an effort strategy $\vec{\beta}(q)=(\beta_{j,k}(q))_{j\in[m],k\in[n]}$, where each contestant $i$ determines her effort as $e_i=\beta_{J_i,k_{J_i}}(q_i)$.

We begin by characterizing the effort strategy in this variant model. Suppose contestant $i$ chooses contest $j\in[m]$. Under the cumulative choice strategy $\vec{\Phi}(q)$, any other contestant $i'$ will also select contest $j$ with a probability $\Phi_j(1)$. Consequently, the number of competitors (excluding contestant $i$) in contest $j$, namely $k_j-1$, follows a binomial distribution $B(\Phi_j(1),n-1)$. Furthermore, the conditional probability that contestant $i'$ is ranked before contestant $i$ given $J_{i'}=j$ is $\Pr[q_{i'}\leq q_i|J_{i'}=j]={\Phi_j(q_i)}/{\Phi_j(1)}$. Therefore, conditioning on that $k_j=k$, the distribution of $\mathrm{rank}(i,\boldsymbol{J},\boldsymbol{e})-1$ is $B(\frac{\Phi_j(q_i)}{\Phi_j(1)},k-1)$. Similar to \Cref{lemma:contestant-equilibrium-effort-condition}, we can derive the effort strategy as %Similar with \Cref{lemma:contestant-equilibrium-effort-condition}, we know that the effort strategy is given by
% We first characterize the effort strategy in this variant model. Consider that contestant $i$ chooses the contest $j\in[m]$. Under a cumulative choice strategy $\vec{\Phi}(q)$, any other contestant $i'$ also chooses contest $j$ with a probability $\Phi_j(1)$. Thus, the number of competitors (exclude contestant $i$) in contest $j$ is $k_j-1$, which follows the binomial distribution $B(\Phi_j(1),n-1)$. Furthermore, we can calculate the conditional probability that $i'$ is ranked before $i$ when $J_{i'}=j$ is $\Pr[q_{i'}\leq q_i|J_{i'}=j]={\Phi_j(q_i)}/{\Phi_j(1)}$. Therefore, conditioning on $k_j=k$, the distribution of contestant $i$'s rank, $\mathrm{rank}(i,\boldsymbol{J},\boldsymbol{e})-1$, is also a binomial distribution $B(\frac{\Phi_j(q_i)}{\Phi_j(1)},k-1)$.
\begin{equation}
    \beta_{j,k}(q_i)=\int_{q_i}^1v(t)(-\frac{dx_{\vec{w}_j}^{(k)}(\frac{\Phi(t)}{\Phi(1)})}{dt})dt,\label{eq:betajk}
\end{equation}
where $x_{\vec{w}_j}^{(k)}(\phi)$ is defined as
\[
    x_{\vec{w}_j}^{(k)}(\phi):=\sum_{l=1}^{k}w_{j,l}\binom{k-1}{l-1}\phi^{l-1}(1-\phi)^{k-l}.
\]
Note that $x_{\vec{w}_j}^{(k)}(\frac{\Phi_j(q_i)}{\Phi_j(1)})$ represents the conditional expectation of the prize obtained by contestant $i$ when $J_i=j$ and $k_j=k$. Interestingly, although the value of $x_{\vec{w}_j}^{(k)}(\frac{\Phi_j(q_i)}{\Phi_j(1)})$ varies with $k$, a key observation is that the expected prize amount for contestant $i$ from contest $j$ remains unchanged, which is formally stated in the following lemma:
\begin{lemma}\label{lemma:known-competitor-number-same-prize}
    For each $j \in [m]$, it holds that $\E_{k}[x_{\vec{w}_j}^{(k)}(\frac{\Phi_j(q_i)}{\Phi_j(1)})]$ $=x_{\vec{w}_j}(\Phi_j(q_i))$, where $k-1\sim B(n-1,\Phi_j(1))$.
\end{lemma}
% \begin{proof}
% Define $p=\Phi_j(1)$ and $q=\Phi_j(q_i)$, we prove that
% \begin{align*}
%     &\sum_{k=1}^n\binom{n-1}{k-1}p^{k-1}(1-p)^{n-k}\sum_{l=1}^{k}w_{j,l}\binom{k-1}{l-1}(\frac{q}{p})^{l-1}(1-\frac{q}{p})^{k-l}
%     \\=&\sum_{k=1}^nw_{j,k}\binom{n-1}{k-1}q^{k-1}(1-q)^{n-k}.
% \end{align*}
% For convenience, we rewrite $n-1$ as $n$, $k-1$ as $k$, and $l-1$ as $l$, then we need to prove \begin{align*}
%     &\sum_{k=0}^{n}\binom{n}{k}p^{k}(1-p)^{n-k}\sum_{l=0}^{k}w_{j,l+1}\binom{k}{l}(\frac{q}{p})^{l}(1-\frac{q}{p})^{k-l}
%     \\=&\sum_{k=0}^nw_{j,k+1}\binom{n}{k}q^{k}(1-q)^{n-k}.
% \end{align*}
% We prove that the two sides have equal coefficient for each $w_{j,l+1}$, as follows: \begin{align*}
% &\binom{n}{l}q^{l}(1-q)^{n-l}
% \\=&\binom{n}{l}p^{l}(\frac{q}{p})^{l}(p-q+1-p)^{n-l}
% \\=&\binom{n}{l}p^{l}(\frac{q}{p})^{l}\sum_{t=0}^{n-l}\binom{n-l}{t}(p-q)^t(1-p)^{n-l-t}
% \\=&\binom{n}{l}p^{l}(\frac{q}{p})^{l}\sum_{k=l}^{n}\binom{n-l}{k-l}(p-q)^{k-l}(1-p)^{n-k}
% \\=&\sum_{k=l}^{n}\binom{n}{n-l}\binom{n-l}{n-k}p^{k}(1-\frac{q}{p})^{k-l}(1-p)^{n-k}(\frac{q}{p})^{l}
% \\=&\sum_{k=l}^{n}\binom{n}{n-k}\binom{k}{l}p^{k}(1-p)^{n-k}(\frac{q}{p})^{l}(1-\frac{q}{p})^{k-l}.
% \end{align*}
% This completes the proof.\qed
% \end{proof}

By \Cref{lemma:known-competitor-number-same-prize} and \Cref{eq:betajk}, we know that when $k-1\sim B(n-1,\Phi_j(1))$, it holds that
\begin{align*}
    \E_k[\beta_{j,k}(q_i)]=&\int_{q_i}^1v(t)(-\frac{d\E_k[x_{\vec{w}_j}^{(k)}(\frac{\Phi(t)}{\Phi(1)})]}{dt})dt
    \\=&\int_{q_i}^1v(t)(-\frac{d x_{\vec{w}_j}(\Phi(t))}{dt})dt =\beta_j(q_i).
\end{align*}
This implies that in the variant model, given $\vec{\Phi}(q)$, the contestant $i$'s expected utility from choosing contest $j$ in the first substage is exactly $v(q_i)x_j(\Phi_j(q_i))-\beta_j(q_i)$, which matches the expected utility in the original model. Consequently, the cumulative choice strategy in an sBNE is still characterized by \Cref{theorem:contestant-equilibrium-wine22} and \Cref{theorem:contestant-equilibrium-wine22}. Moreover, the utility of designers are also unchanged, for either effort objective or participation objective, which suggests that our subsequent results on prize structure design are also applicable to this variant model.

%Revisiting this counter-intuitive conclusion, we elaborate the reason behind it. In the original model, when contestants select contests and exert effort simultaneously, we know that one contestant eventually cares about how many other contestants with a higher skill choose the same contest as her. On the other hand, when they make decisions sequentially and the number of competitors are known, one contestant considers about the number of contestants with a higher skill in these competitors. Because realizing the number of competitors depends on the cumulative choice strategy, the above two model are equivalent in expectation. This also suggests that our subsequent results on prize structure design are also applicable to this variant model. %In summary, this conclusion gives us insights on both theory and reality

%In summary, in the variant model where contestants' efforts may depend on the number of competitors, under the contestant equilibrium, the choice of contests remains the same as in the original model, and a contestant's effort also equals to that in the original model in expectation. Consequently, the utility of designers are also unchanged, for either effort objective or participation objective. This suggests that our subsequent results on prize structure design are also applicable to this variant model.

\section{Optimal Prize Structure}\label{sec:designer}
% In the last section, we focus on the aspect of contestants (i.e., the second stage of our game model) and discuss the contestant equilibrium when all prize structures are given. In this section, based on the results of contestant equilibrium, we move back to the first stage of our game model and mainly explore the optimal strategies of designers (i.e., prize structure design) under two kinds of objective functions. For the effort objective, we show that under a reasonable assumption which covers a wide range of objectives, the winner-takes-all prize structure is the dominating strategy for every designer, and thus constitutes the subgame perfect equilibrium (SPE). For the participation objective, we show that in the equilibrium every designer adopts a simple contest where the budget is equally distributed to a fixed number of contestants. We give an efficient algorithm to compute such an equilibrium.
In the last section, we have characterized the equilibrium of contestants in the second stage given all prize structures. 
In this section, we move back to the first stage of the game model, and explore the optimal strategies for designers (i.e., prize structure design) under two kinds of different objectives. 
%For the effort objective, we demonstrate that, for a broad range of reasonable objectives, the winner-takes-all prize structure is the dominant strategy for every designer, thereby constituting the subgame perfect equilibrium (SPE). For the participation objective, we show that the optimal strategy for every designer is always a simple contest where the budget is evenly distributed among a fixed number of contestants, and the equilibrium can be computed efficiently when all designers have the same threshold for eligible contestants.

Before examining the optimal prize structure, we introduce some symbols to help us simplify the objective functions. For a contest designer $j\in[m]$, let $\Phi_j^*(q;x_1(\phi), \cdots, $ $x_m(\phi))$ denote the $j$-th component of the cumulative choice strategy in the contestant equilibrium as given by \Cref{theorem:contestant-equilibrium-wine22}. Note that $\Phi_j^*(q;x_1(\phi),\cdots,x_m(\phi))$ is uniquely determined as long as $x_j(\phi)$ is strictly decreasing (or equivalently, $w_{j,1}>w_{j,n}$). Specifically, for all $q\in[0,1]$, 

$$\Phi_j^*(q;x_1(\phi),\cdots,x_m(\phi))=x_j^{-1}(Q^{-1}(q)).$$

Recall that $Q^{-1}(q)$ is the inverse function of $Q(x)=\sum_{j'\in[m]}x_{j'}^{-1}(x)=x_j^{-1}(x)+\sum_{j'\neq j}x_{j'}^{-1}(x).$ When analyzing designer $j$'s strategy, the effect of other contests is captured by the term $\sum_{j'\neq j}x_{j'}^{-1}(x)$. For convenience, we denote $x_{-j}^{-1}(x):=\sum_{j'\neq j}x_{j'}^{-1}(x)$. We use the notation $\Phi_j^*(q;x_j(\phi),x_{-j}^{-1}(x))$ or simply $\Phi_j^*(q)$ to refer to $\Phi_j^*(q;x_1(\phi),\cdots,x_m(\phi))$ when there is no ambiguity.

Remark that a subtlety arises when $w_{j,1}=w_{j,n}$, meaning that $x_j(\phi)$ is constant for $\phi\in[0,1]$, which leads to potential non-uniqueness of $\Phi_j^*(\theta_j;x_j(\phi),x_{-j}^{-1}(x))$ by \Cref{theorem:contestant-equilibrium-wine22}. In such case, a reasonable selection rule for contestant equilibrium is equal division: If $x_j(\phi)$ and $x_{j'}(\phi)$ are constant and equal for some $j, j'\in[m]$, we assume $\Phi_j^*(q;x_1(\phi),\cdots,x_m(\phi))=\Phi_{j'}^*(q;x_1(\phi),\cdots,x_m(\phi))$ for all $q\in[0,1]$. In other words, identical contests always attract contestants equally. 

To characterize the optimal contests, we define the concept of a simple contest, as studied in the literature \cite{EGG21}.
\begin{definition}
    A prize structure $\vec{w}_j$ is called a simple contest, if there exists $k\in[n]$, such that the prize is equally allocated to the first $k$ contestants, i.e., $w_{j,1}=\cdots=w_{j,k}>0$, and $w_{j,k+1}=\cdots=w_{j,n}=0$.
    
    Denote $\vec{w}_j^{(k,T)}$ by the simple contest with $k$ positive prizes dividing the budget of $T$, i.e., $w_{j,1}^{(k,T)}=\cdots=w_{j,k}^{(k,T)}=\frac{T}{k}$.
\end{definition}

Consider the simple contest having $k$ positive prizes with a total budget of $1$. We denote its corresponding interim allocation function by

$$\xi_k(\phi):=x_{\vec{w}_j^{(k,1)}}(\phi)=\frac1k\sum_{l=1}^k\binom{n-1}{l-1}\phi^l(1-\phi)^{n-l}.$$

% \lny{not used For any prize structure $\vec{w}_j$, the interim allocation function $x_{\vec{w}_j}(\phi)$ can be viewed as a combination of $\xi_1(\phi),\cdots,\xi_n(\phi)$:
% $$x_{\vec{w}_j}(\phi)=\sum_{k=1}^{n}k(w_{j,k}-w_{j,k+1})\xi_k(\phi).$$
% Here we define $w_{j,n+1}=0$.}

% \lny{For convenience, we assume $w_{j,1}>w_{j,n}$ for all $j\in[m]$ in this section, so that . }
%\lny{too long}
\subsection{Effort Objective}
When contest designer $j$'s utility is determined by an effort objective function specified by $\vec{\alpha}_j$, her utility depends on the efforts of participants at each rank. We first derive a convenient expression for designer $j$'s utility.
\begin{lemma}\label{lemma:designer-utility-effort-objective}
    Given $x_j(\phi),x_{-j}^{-1}(x)$ and $\vec{\alpha}_j$, denote designer $j$'s utility as $\hat{R}_j(x_j(\phi),$ $x_{-j}^{-1}(x),\vec{\alpha}_j)$, then we have
    \begin{align*}
        \hat{R}_j&(x_j(\phi),x_{-j}^{-1}(x),\vec{\alpha}_j)\\
    =&\int_{0}^{+\infty}v(x_j^{-1}(x)+x_{-j}^{-1}(x))G(x_j^{-1}(x);\vec{\alpha}_j)dx.
    \end{align*}
    Here we define $v(q)=0$ for all $q>1$, and define 
    $$G(\phi;\vec{\alpha}_j):=n\sum_{k=1}^n\alpha_{j,k}\int_0^{\phi} g_k(t)dt$$
    with $g_k(\phi):=\Pr[\mathrm{rank}(i,\boldsymbol{J},\boldsymbol{e})=k|\Phi_j^*(q_i)=\phi]=\binom{n-1}{k-1}\phi^{k-1}(1-\phi)^{n-k}$.
\end{lemma}

We say that a coefficient vector $\vec{\alpha}_j$ is \emph{weight-monotone}, if $G(\phi;\vec{\alpha}_j)$ is non-decreasing and concave in $\phi\in[0,1]$, which is equivalent to that  $\sum_{k=1}^n\alpha_{j,k}g_k(\phi)$ is non-negative and non-increasing in $\phi$. The following lemma provides a sufficient condition for weight-monotonicity.

\begin{lemma}\label{lemma: sufficient condition on weight-monotone}
    If $\vec{\alpha}_j$ is non-increasing and non-negative, i.e., $\alpha_{j,1}\geq \alpha_{j,2}\geq\cdots\geq\alpha_{j,n}\geq 0$, then $\vec{\alpha}_j$ is weight-monotone.
\end{lemma}
% \begin{proof}
%     Observe that $\sum_{k=1}^n\alpha_{j,k}g_k(\phi)=\sum_{k=1}^{n}k(\alpha_{j,k}-\alpha_{j,k+1})\xi_k(\phi)$ where $\alpha_{j,n+1}$ is defined as $0$. Since $\xi_k(\phi)$ is non-increasing on $[0,1]$ and $(\alpha_{j,k}-\alpha_{j,k+1})$ is non-negative for all $k\in[n]$, $\sum_{k=1}^n\alpha_{j,k}g_k(\phi)$ is non-increasing in $\phi\in[0,1]$. Thus $\vec{\alpha}_j$ is weight-monotone by definition.
%     \qed
% \end{proof}

Notably, by \Cref{lemma: sufficient condition on weight-monotone}, weight-monotone effort objectives encompass many widely-studied objectives in contest theory, such as total effort \cite{MS01}, maximum effort \cite{CHS19}, or total effort of the top $k$ contestants \cite{AS09}.
%Note that most effort-based objectives such as the total effort \cite{MS01}, the maximum effort \cite{CHS19} or the total effort of top $k$ contestants \cite{AS09}, can belong to the class of weight-monotone effort objective. In another word, weight-monotone effort objective covers most of mainstream objectives in contest theory. %be represented as a weight-monotone vector...

We now analyze the optimal prize structure $\vec{w}_j$ (i.e., the best response) given $x_{-j}^{-1}(x)$, utilizing the single-crossing property \cite{MOLDOVANU200670} of the interim allocation functions induced by the prize structure.
\begin{definition}\label{def:single-crossing}
For two functions $f(\phi)$ and $g(\phi)$ on $[0,1]$, we say $f(\phi)$ is single-crossing with respect to $g(\phi)$, if there exists some point $\phi_0\in[0,1]$, such that $f(\phi)\geq g(\phi)$ for all $\phi\in[0,\phi_0)$, and $f(\phi)\leq g(\phi)$ for all $\phi\in(\phi_0,1]$.

For two interim allocation functions $x_j(\phi)$ and $\tilde{x}_j(\phi)$, we say $x_j(\phi)$ single-crossing-dominates $\tilde{x}_j(\phi)$, if the following three conditions hold:
\begin{enumerate}
    \item $\int_0^1x_j(\phi)d\phi\geq \int_0^1\tilde{x}_j(\phi)d\phi$.
    \item $x_j(\phi)$ is single-crossing with respect to $\tilde{x}_j(\phi)$.
    \item $\frac{-dx_j(\phi)}{d\phi}$ is single-crossing with respect to $\frac{-d\tilde{x}_j(\phi)}{d\phi}$.
\end{enumerate}
    % For two interim allocation functions $x_j(\phi)$ and $\tilde{x}_j(\phi)$, we say $x_j(\phi)$ is single-crossing with respect to $\tilde{x}_j(\phi)$, if there exists some point $\phi_0\in[0,1]$ such that $x_j(\phi)\geq \tilde{x}_j(\phi)$ for all $\phi\in[0,\phi_0)$, and $x_j(\phi)\leq \tilde{x}_j(\phi)$ for all $\phi\in(\phi_0,1]$.
\end{definition}

Intuitively, if $x_j(\phi)$ is single-crossing with respect to $\tilde{x}_j(\phi)$, then $x_j(\phi)$ offers a prize structure that is more attractive to high-skill contestants, but potentially less attractive to low-skill contestants, compared to $\tilde{x}_j(\phi)$. Furthermore, if $x_j(\phi)$ single-crossing-dominates $\tilde{x}_j(\phi)$, the former prize structure will guarantee to outperform the latter regardless of the strategies of other designers, as stated in the following theorem.
% in Theorem \ref{thm: compare under single-crossing}.

\begin{theorem}\label{thm: compare under single-crossing}
    Assuming $\vec{\alpha}_j$ is weight-monotone, for any two prize structures $\vec{w}_j$ and $\vec{w}'_j$ such that $x_{\vec{w}_j}(\phi)$ single-crossing-dominates $x_{\vec{w}_j'}(\phi)$, it holds for any $x_{-j}^{-1}(x)$ that
    $$\hat{R}_j(x_{\vec{w}_j}(\phi),x_{-j}^{-1}(x),\vec{\alpha}_j)\geq \hat{R}_j(x_{\vec{w}_j'}(\phi),x_{-j}^{-1}(x),\vec{\alpha}_j).$$
\end{theorem}

According to Theorem \ref{thm: compare under single-crossing}, if $x_{\vec{w}_j}(\phi)$ single-crossing-dominates $x_{\vec{w}_j'}(\phi)$, then $\vec{w}_j$ will yield higher utility for designer $j$. This implies that if we can construct a $\vec{w}_j$ such that $x_{\vec{w}_j}(\phi)$ single-crossing-dominates any $x_{\vec{w}_j'}(\phi)$, then $\vec{w}_j$ is a dominant strategy and thus the optimal prize structure. Based on this argument, we show that the winner-take-all prize structure is optimal, as demonstrated in Theorem \ref{thm: winner-take-all dominant}.
\begin{theorem}\label{thm: winner-take-all dominant}
    Assuming $\vec{\alpha}_j$ is weight-monotone, then the winner-take-all prize structure $\vec{w}_j^*=\vec{w}_j^{(1,T_j)}$ is a dominant strategy for designer $j$. That is, for any $x_{-j}^{-1}(x)$ and any feasible prize structure $\vec{w}'_j$, it holds that
    $\hat{R}_j(x_{\vec{w}_j^*}(\phi),x_{-j}^{-1}(x),\vec{\alpha}_j$ $\geq \hat{R}_j(x_{\vec{w}_j'}(\phi),x_{-j}^{-1}(x),\vec{\alpha}_j).$
\end{theorem}
% \begin{proof}
%     We only need to prove that $\vec{w}_j^*$ is single-crossing w.r.t. $\vec{w}'_j$.
%     \notqed
% \end{proof}

% Since winner-take-all prize structure is a dominant strategy for any designer, as a corollary, they form a unique equilibrium among designers. Combining the corresponding contestant equilibrium, it is also the unique SPE of our two-stage game model.  
Since the winner-take-all prize structure is a dominant strategy for any designer, it follows as a corollary that it forms an equilibrium among designers. Combining this with the corresponding contestant equilibrium, we can obtain an SPE of our two-stage game model.
\begin{corollary}\label{cor:effort-objective-designer-SPE}
    Assuming $\vec{\alpha}_j$ is weight-monotone for all designers $j\in[m]$, there is an SPE where
    every designer implements the winner-take-all prize structure with all budget.
    % Moreover, this SPE is unique if we ignore the designers who get zero utility.
\end{corollary}

\subsection{Participation Objective}
% In this subsection, we concentrate on the participation objective, which does not depend on the participants' effort and care about the number of eligible participants. For example, the objective in \cite{EGG21}, the number of participants whose values achieve a threshold, belongs to this class.
In this subsection, we focus on the participation objective, which is concerned with the number of eligible participants rather than their efforts.
%For instance, the objective in \cite{EGG21}, which counts the number of participants whose values meet a specified threshold, belongs to this category.

% To characterize the optimal contests with respect to participation objective, we define the concept of a simple contest, as studied in the literature \cite{EGG21}.
% \begin{definition}
%     A prize structure $\vec{w}_j$ is called a simple contest, if there exists $k\in[n]$, such that the prize is equally allocated to the first $k$ contestants, i.e., $w_{j,1}=\cdots=w_{j,k}>0$, and $w_{j,k+1}=\cdots=w_{j,n}=0$.
    
%     Denote $\vec{w}_j^{(k,T)}$ by the simple contest with $k$ positive prizes dividing the budget of $T$, i.e., $w_{j,1}^{(k,T)}=\cdots=w_{j,k}^{(k,T)}=\frac{T}{k}$.
% \end{definition}

% Consider the simple contest having $k$ positive prizes with a total budget of $1$. We denote its corresponding interim allocation function by

% $$\xi_k(\phi):=x_{\vec{w}_j^{(k,1)}}(\phi)=\frac1k\sum_{l=1}^k\binom{n-1}{l-1}\phi^l(1-\phi)^{n-l}.$$

When designer $j$ has a participation objective specified by $\theta_j$, let $\hat{R}_j(x_j(\phi),x_{-j}^{-1}(x),\theta_j)$ denote her expected utility given $x_j(\phi)$ and $x_{-j}^{-1}(x)$, then we have

$$\hat{R}_j(x_j(\phi),x_{-j}^{-1}(x),\theta_j)= n \cdot\Phi_j^*(\theta_j;x_j(\phi),x_{-j}^{-1}(x)).$$

Intuitively, to attract more eligible participants, the designer should focus on her contest's appeal to those contestants at the threshold, which is always achieved by a simple contest, as shown in the following.
% it is unnecessary to vary prize amounts within the prize structure. Instead, we should focus on balancing a single prize amount with the number of prizes allowed within the budget, since the objective treats two eligible participants with different skills as equivalent. Consequently, the optimal prize structure for the participation objective is as follows:
\begin{theorem}\label{thm: participant objective- simple contest}
    Under the participation objective specified by $\theta_j$, given any $x_{-j}^{-1}(x)$, the optimal utility for designer $j$ is achieved by a simple contest. 
    Specifically, there is an optimal simple contest $\vec{w}_{j}^{(k^*,T_j)}$ with $k^*\in[n]$ positive prizes, such that
    
    $$\hat{R}_j(x_{\vec{w}_{j}^{(k^*,T_j)}}(\phi),x_{-j}^{-1}(x),\theta_j)=\max_{\vec{w}_j}\hat{R}_j(x_{\vec{w}_{j}}(\phi),x_{-j}^{-1}(x),\theta_j).$$
    Moreover, let $\phi^*=\Phi_j^*(\theta_j;x_{\vec{w}_{j}^{(k^*,T_j)}}(\phi),x_{-j}^{-1}(x))$, then $k^*$ satisfies
    $$k^*=\arg\max_{k}\xi_{k}(\phi^*).$$
\end{theorem}

By Theorem \ref{thm: participant objective- simple contest}, given the strategies of other designers, designer $j$'s best response is to use a simple contest and determine the number of prizes $k$. %By the definition of equilibrium, designers take their own best response to each other. 
To find a designer equilibrium, we need to identify each designer's value of $k$ such that all designers are implementing their best response strategies. The following theorem demonstrates that if all designers have the same threshold, there is an efficient algorithm to determine the designer equilibrium.
\begin{theorem}\label{thm:participant objective- spe}
    When all designers have a common participation objective specified by $\theta_j=\theta$ for a common threshold $\theta\in[0,1]$, %\lny{different threshold?}, 
    the SPE exists and can be computed efficiently.
\end{theorem}
% \begin{proof}
% Define $\bar{x}_j(\phi):=T_j\max_{k\in[n]}\xi_k(\phi)$.
%     \notqed
% \end{proof}

\section{Conclusion and Future Works}
In this paper, we analyze a game model involving multiple rank-order contests held in parallel. Contestants select a contest and exert costly effort to win better prizes, while contest designers optimize prize structures for certain objectives. Given fixed prize structures, we provide a detailed characterization of contestant equilibrium. Additionally, we propose optimal prize structures for designers under specific assumptions for both effort and participation objectives.

We suggest two directions for future research. First, an interesting question is the behavior of contestants when they can attend in more than one contests. Second, generalizing our model to general cost function of contestants and other designer objectives may yield interesting results. %Additionally, explaining the consistency of optimal prize structure between single-contest and multiple-contest settings may provide further understanding of multiple-contest environments.

\clearpage
\newpage
\bibliographystyle{named} 
\bibliography{main}

\newpage
\appendix
\section*{Appendix}

\section{Missing Proofs in Section \ref{sec:CE}}

\subsection{Proof of Lemma \ref{lemma:SBNE-deterministic-effort}}
\begin{proof}
If $w_{j,1}=w_{j,n}$, then for all $q_i\in[0,1]$ and $e\geq 0$, we have $\bar{u}(j,e;\tau(q);q_i)=v(q_i)w_{j,1}-e$. This implies that any strategy $(j,e)$ with $e>0$ is strictly dominated by $(j,0)$. Therefore, the sBNE $\tau(q)$ must satisfy  $\Pr_{(J_i,e_i)\sim \tau(q_i)}[e_i=0|J_i=j]=1$. For the subsequent analysis, we assume $w_{j,1}>w_{j,n}$, and thus $x_j(\phi)$ is strictly decreasing.

Let $\tau_{e_i|j}(q_i)$ denote the distribution of $e_i$ conditioning on $J_i=j$ when $(J_i,e_i)\sim \tau(q_i)$. Let $\bar{\tau}_{e_i|j}$ denote the marginal distribution conditioning on $J_i=j$ when $q_i\sim U[0,1]$ and $(J_i,e_i)\sim \tau(q_i)$.
We first prove the following claims:
\begin{claim}\label{claim.lem1.1}
$\bar{\tau}_{e_i|j}$ is supported on some interval $[0,\bar{e}]$ or $[0,+\infty)$ with positive density and no mass point. In other words, $\Psi_j^{\tau(q)}(e)$ is strictly decreasing and continuous in $e$ on $[0,\bar{e}]$ (or $[0,+\infty)$, respectively), with $\Psi_j^{\tau(q)}(0)=\pi_j(1)=\Pr_{(J_i,e_i)\sim \tau(q_i)}[J_i=j]$ and $\Psi_j^{\tau(q)}(\bar{e})=0$ (or $\lim_{e_i\to+\infty}\Psi_j^{\tau(q)}(e_i)=0$, respectively).
\end{claim}
\begin{proof}
The proof of this claim resembles standard arguments in auction theory.
% Firstly, $\bar{\tau}_{e_i|j}$ has a bounded support because obviously any $e_i>v(0)w_{j,1}$ is strictly dominated by $e_i=0$. We can take $\bar{e}=\inf\{e_i:\Psi_j^{\tau(q)}(\bar{e})=0\}$.
Firstly, we take $\bar{e}=\inf\{e_i:\Psi_j^{\tau(q)}(\bar{e})=0\}$. In case that $\bar{e}=+\infty$, we still use the notation $[0,\bar{e}]$ to denote the support $[0,+\infty)$. 
%\lny{no need?}
% For a contestant with quantile $q_i$, any $e_i>v(q_i)w_{j,1}$ incurs negative utility, and is therefore strictly dominated by $e_i=0$. Therefore we have $\Psi_j^{\tau(q)}(e_i)\leq \Pr_{q_i\sim U[0,1]}[v(q_i)\geq e_i]$, and consequently $\lim_{e_i\to+\infty}\Psi_j^{\tau(q)}(e_i)=0$.

Next, we prove that $\bar{\tau}_{e_i|j}$ has no mass point. Suppose for contradiction that $e_1$ is a mass point of 
$\bar{\tau}_{e_i|j}$, i.e., 

$$\lim_{e_i\to e_1-0}\Psi_j^{\tau(q)}(e_i)>\lim_{e_i\to e_1+0}\Psi_j^{\tau(q)}(e_i).$$ 
Take 

$$\delta=x_j(\lim_{e_i\to e_1+0}\Psi_j^{\tau(q)}(e_i))-x_j(\lim_{e_i\to e_1-0}\Psi_j^{\tau(q)}(e_i)),$$ 
then it holds that $\delta>0$ since $x_j(\phi)$ is strictly decreasing. By the tie-breaking rule, there exists some quantile $q_i$ such that a contestant with $q_i$ plays $e_1$ with positive probability, and that her expected prize by playing $e_1$ is less than $x_j(\lim_{e_i\to e_1+0}\Psi_j^{\tau(q)}(e_i))-\frac12\delta$. For such contestants, the effort $e_1$ is strictly dominated by $e_1+\epsilon$ for arbitrary $\epsilon\in(0,\frac12{\delta}v(q_i))$, which contradicts.

Then, we prove that $\bar{\tau}_{e_i|j}$ have positive probability mass in every subinterval of $[0,\bar{e}]$, i.e., $\Psi_j^{\tau(q)}(e)$ is strictly decreasing. Suppose for contradiction that there is $(e_1,e_2)$ with $0\leq e_1<e_2\leq \bar{e}$ such that $\Pr_{e_i\sim\tau_{e_i|j}(q_i)}[e_i\in(e_1,e_2)]=0$. Without loss of generality we take $e_2$ such that 

$$e_2=\sup\{e_i\in(e_1,+\infty):\Pr_{e_i\sim\tau_{e_i|j}(q_i)}[e_i\in(e_1,e_2)]=0\}.$$ 
We know $e_2<+\infty$ because otherwise $\bar{e}\leq e_1$. Now we have $\Pr_{e_i\sim\tau_{e_i|j}(q_i)}[e_i\in[e_2,e_2+\epsilon)]>0$ for any $\epsilon>0$. However, since $\bar{\tau}_{e_i|j}$ has no mass point, we have $\lim_{e_i\to e_2+0}\Psi_j^{\tau(q)}(e_i)=\Psi_j^{\tau(q)}(e_1)$. By the continuity of $x_j$, we have $\lim_{e_i\to e_2+0}x_j(\Psi_j^{\tau(q)}(e_i))=x_j(\Psi_j^{\tau(q)}(e_1))$. Take a sufficiently small $\epsilon'\in(0,\epsilon)$, such that $\Psi_j^{\tau(q)}(e_2)>\Psi_j^{\tau(q)}(e_2+\epsilon')>\Psi_j^{\tau(q)}(e_2+\epsilon)$. For any $q_i$ such that $\Pr_{e_i\sim \tau_{e_i|j}(q_i)}[e_i\in [e_2,e_2+\epsilon')]>0$, since a contestant with $q_i$ has no incentive to deviate to $e_2+\epsilon$, we have 

$$v(q_i)x_j(\Psi_j^{\tau(q)}(e_2+\epsilon'))-e_2\geq v(q_i)x_j(\Psi_j^{\tau(q)}(e_2+\epsilon))-(e_2+\epsilon),$$ 
which implies 

$$v(q_i)\leq \frac{\epsilon-\epsilon'}{x_j(\Psi_j^{\tau(q)}(e_2+\epsilon))-x_j(\Psi_j^{\tau(q)}(e_2+\epsilon'))}=:M.$$
Therefore, for sufficiently small $\epsilon''\in(0,\epsilon')$, it holds that, for all $e_i\in[e_2,e_2+\epsilon'')$, 
\begin{align*}
    &v(q_i)x_j(\Psi_j^{\tau(q)}(e_i))-e_i-(v(q_i)x_j(\Psi_j^{\tau(q)}(e_1))-e_1)\\
    \leq &M\cdot(x_j(\Psi_j^{\tau(q)}(e_i))-x_j(\Psi_j^{\tau(q)}(e_1)))-e_2+e_1<0,
\end{align*}
i.e., $e_i$ is strictly dominated by $e_1$, which contradicts. 

This completes the proof of this claim.
\end{proof}
The next claim tells that under an sBNE, the exerted effort of contestant with higher ability in any contest will be not less than that of contestant with lower ability.
\begin{claim}\label{claim.lem1.2}
    For any $0\leq q_1< q_2\leq 1$ and any contest $j$, it holds that $\Pr_{e_1\sim\tau_{e_i|j}(q_1),e_2\sim\tau_{e_i|j}(q_2)}[e_1\geq e_2]=1$.
\end{claim}
\begin{proof}
    Fix any $q_1<q_2$, let $e_1$ and $e_2$ be independently drawn from $\tau_{e_i|j}(q_1)$ and $\tau_{e_i|j}(q_2)$, respectively. Denote $X_{e_1}=x_j(\Psi_j^{\tau(q)}(e_1))$ and $X_{e_2}=x_j(\Psi_j^{\tau(q)}(e_2))$. By the definition of sBNE, with probability $1$, both of the following inequalities hold:
    
    $$v(q_1)X_{e_1}-e_1\geq v(q_1)X_{e_2}-e_2$$
    and 
    $$v(q_2)X_{e_2}-e_2\geq v(q_2)X_{e_1}-e_1.$$ 
    It implies that
    \begin{align*}
        v(q_1)(X_{e_2}-X_{e_1})\leq e_2-e_1\leq v(q_2)(X_{e_2}-X_{e_1}).
    \end{align*}
% We discuss two cases:
% (a) $q_1=q_2$. In this case, we have $v(q_1)(X_{e_2}-X_{e_1})=e_2-e_1$.

Since $q_1<q_2$, we have $v(q_1)>v(q_2)$, and it must hold that $X_{e_2}-X_{e_1}\leq 0$. As $x_j(\phi)$ and $\Psi_j^{\tau(q)}(e_i)$ are both strictly decreasing, this implies $e_2\leq e_1$.
\end{proof}

Furthermore, we show that the ``equality'' cannot be reached. 
\begin{claim}\label{claim.lem1.3}
    For any $0\leq q_1< q_2\leq 1$, if $\Phi_{\pi_j}(q_1)<\Phi_{\pi_j}(q_2)$, or equivalently, $\Pr[q_i\leq q_1\land J_i=j]<\Pr[q_i\leq q_2\land J_i=j]$, then $\Pr_{e_1\sim\tau_{e_i|j}(q_1),e_2\sim\tau_{e_i|j}(q_2)}[e_1>e_2]=1$.
\end{claim}
\begin{proof}
We only need to prove $\Pr[e_1=e_2]=0$. Observe that for any $q_3\in(q_1,q_2)$ and $e_3\sim \tau_{e_i|j}(q_3)$, we have already known that $e_1\geq e_3\geq e_2$ with probability $1$. Suppose for contradiction that $\Pr[e_1=e_2]>0$, then there exists some $e_0$, such that for any $q_3\in(q_1,q_2)$, and $\Pr_{e_3\sim \tau_{e_i|j}(q_3)}[e_3=e_0]=1$. This implies that $\bar{\tau}_{e_i|j}$ has a point of mass at least $q_2-q_1$ at $e_0$, contradicting with that of $\bar{\tau}_{e_i|j}$ has no mass point.
\end{proof}

Now, we state that in an sBNE, the probability that one contestant joins contest and exert an effort greater than $e_1$ is equal to the probability that one contestant with quantile less than $q_1$ joins contest $j$.
\begin{claim}\label{claim.lem1.4}
For any $q_1\in[0,1]$ such that $\pi_j(q_1)>0$, it holds that $\Pr_{e_1\sim\tau_{e_i|j}(q_1)}[\Psi_j^{\tau(q)}(e_1)=\Phi_{\pi_j}(q_1)]=1$.
\end{claim}
\begin{proof}
Let $e_1$ be drawn from $\tau_{e_i|j}(q_1)$ and let $q_2\sim U[0,1]$ and $e_2\sim \tau_{e_i|j}(q_2)$. Observe that the event $\{\Phi_{\pi_j}(q_1)\neq\Phi_{\pi_j}(q_2)\}$ happens with probability $1$, and conditioning on this, by \Cref{claim.lem1.3}, we know that $q_2>q_1$ if and only if $e_2<e_1$. Therefore, we obtain that $\Pr_{e_1\sim \tau_{e_i|j}(q_1)}[\Psi_j^{\tau(q)}(e_1)=\Phi_{\pi_j}(q_1)]=1$.
\end{proof}

Finally, we show that if two quantiles have the same $\pi_j(q)$, they will exert the same effort in contest $j$.
\begin{claim}\label{claim.lem1.5}
For any $q_1$ and $q_2\in[0,1]$ such that $\pi_j(q_1)>0$ and $\pi_j(q_2)>0$, if $\Phi_{\pi_j}(q_1)=\Phi_{\pi_j}(q_2)$, then it holds that $\Pr_{e_1\sim\tau_{e_i|j}(q_1),e_2\sim\tau_{e_i|j}(q_2)}[e_1=e_2]=1$.
\end{claim}
\begin{proof}
Assume, for the sake of contradiction, that the probability $\Pr_{e_1\sim\tau_{e_i|j}(q_1),e_2\sim\tau_{e_i|j}(q_2)}[e_1=e_2]<1$. Let $e_1$ and $e_2$ be independently drawn from $\tau_{e_i|j}(q_1)$ and $\tau_{e_i|j}(q_2)$, respectively. It happens with positive probability that $e_1<e_2$. However, by \Cref{claim.lem1.4} we have $\Psi_j^{\tau(q)}(e_1)=\Phi_{\pi_j}(q_1)=\Phi_{\pi_j}(q_2)=\Psi_j^{\tau(q)}(e_2)$ with probability $1$. Therefore $v(q_2)x_j(\Psi_j^{\tau(q)}(e_1))-e_1>v(q_2)x_j(\Psi_j^{\tau(q)}(e_2))-e_2$, i.e., the contestant with quantile $q_2$ obtains higher utility using $e_1$ than using $e_2$, which contradicts with that $\tau(q)$ is an sBNE.
\end{proof}

For each $j\in[m]$ and all $q_1\in[0,1]$ such that $\pi_j(q_1)>0$, by setting $q_2=q_1$ in \Cref{claim.lem1.5}, we find that $\tau_{e_i|j}(q_1)$ is a one-point distribution. In other words, there exists a function $\beta_j:[0,1]\to\mathbb{R}_{\geq 0}$ such that for all $q_i\in[0,1]$ with $\pi_j(q_i)>0$, it holds that $\Pr_{e_i\sim\tau_{e_i|j}(q_i)}[e_i=\beta_j(q_i)]=1$. Equivalently, $\Pr_{(J_i,e_i)\sim\tau(q_i)}[e_i=\beta_j(q_i)|J_i=j]=1$. 

Moreover, by \Cref{claim.lem1.5}, we have $\beta_j(q_1)=\beta_j(q_2)$ whenever $\Phi_{\pi_j}(q_1)=\Phi_{\pi_j}(q_2)$. This implies that $\beta_j(q_i)$ can be represented as a function of $\Phi_{\pi_j}(q_i)$, denoted by $\beta_j(q_i)=\hat{\beta}_j(\Phi_{\pi_j}(q_i))$. 
Finally, by \Cref{claim.lem1.3}, $\hat{\beta}_j(\phi)$ is strictly decreasing.
By \Cref{claim.lem1.1}, $\hat{\beta}_j(\phi)$ is continuous because otherwise $\Psi_j^{\tau(q)}(e)$ is discontinuous.
\end{proof}

\subsection{Proof of Lemma \ref{lemma:contestant-equilibrium-effort-condition}}

\begin{proof}
We first show the sufficiency of this lemma. 
\subsubsection{Sufficiency:} ~
Assuming the two conditions are satisfied, we show that $\vec{\Phi}(q)$ and $\vec{\beta}(q)$ constitute an sBNE. Given that $\vec{\Phi}(q)$ is a cumulative choice strategy, by \Cref{def:cumulative behavior}, for all $j\in[m]$, $\Phi_j(q)$ is non-decreasing on $[0,1]$, and thus is differentiable almost everywhere. 

We construct $\vec{\pi}(q)$ as follows: for any $q_i\in[0,1]$, if $\Phi_j'(q_i)$ exists for all $j\in[m]$, then let $\pi_j(q_i)=\Phi_j'(q_i)$ for all $j\in[m]$. Otherwise, let $\vec{\pi}_j(q_i)$ be an arbitrary distribution on $\arg\max_{j\in[m]}v(q_i)x_j(\Phi_j(q_i))-\beta_j(q_i)$. Let $E_1$ denote the set of all the former kind of $q_i$, and let $E_2=[0,1]\setminus E_1$ denote all the latter kind of $q_i$. For all $q_i\in E_1$, we have 

$$\sum_{j\in[m]}\pi_j(q_i)=\frac{d\sum_{j\in[m]}\Phi_j(q_i)}{dq_i}=\frac{dq_i}{dq_i}=1,$$ 
so $\vec{\pi}$ is a valid choice strategy. Moreover, as $E_2$ has measure zero, we have that $\vec{\Phi}(q)$ is the cumulation of $\vec{\pi}(q)$. 

Now we only need to verify that $\vec{\pi}(q)$ and $\vec{\beta}(q)$ represent an sBNE. Let $\tau(q)$ denote the mixed strategy represented by $\vec{\pi}(q)$ and $\vec{\beta}(q)$. For all $q_i\in E_2$, the equilibrium condition is already satisfied by the construction of $\vec{\pi}(q_i)$. For all $q_i\in E_1$, we prove that for any $j\in[m]$ with $\Phi'_j(q_i)>0$, it holds that 

$$(j,\beta_j(q_i))\in\arg\max_{J_i,e_i}\bar{u}(J_i,e_i;\tau(q);q_i).$$
We first prove the following claim:
\begin{claim}\label{claim.lem2.6}
    For any $j\in[m]$, the effort strategy satisfies that $\beta_j(q_i)\in \arg\max_{e_i}\bar{u}(j,e_i;\tau(q);q_i)$.
\end{claim}
\begin{proof}
If $w_{j,1}=w_{j,n}$, $x_j(\phi)$ is a constant on $[0,1]$, and therefore $\beta_j(q_i)=0$ for all $q_i\in[0,1]$ by definition, which satisfies this claim. 
Now we assume $w_{j,1}>w_{j,n}$, thus, we know that $x_j'(\phi)<0$ for all $\phi\in(0,1)$. We can write the condition 1 as 

$$\beta_j(q_i)=\int_{t=q_i}^1v(t)(-x_j'(\Phi_j(t)))d\Phi_j(t).$$
We can observe that for all $q_1,q_2\in[0,1]$, $\beta_j(q_1)\leq\beta_j(q_2)$ if and only if $\Phi_j(q_1)\leq \Phi_j(q_2)$. This implies that for any $q_i\in[0,1]$, $\Psi_j^{\tau(q)}(\beta_j(q_i))=\Phi_j(q_i)$. Therefore, for any $q'\in[0,1]$, we have

$$\bar{u}(j,\beta_j(q');\tau(q);q_i)=v(q_i)x_j(\Phi_j(q'))-\beta_j(q').$$

Given any $q_i\in[0,1]$, consider a contestant with quantile $q_i$ choosing $J_i=j$. It is easy to see that any $e_i>\beta_j(0)$ is strictly dominated by $\beta_j(0)$, so we only need to consider $e_i=\beta_j(q')$ for all $q'\in[0,1]$ and check that $\bar{u}(j,\beta_j(q_i);\tau(q);q_i)\geq\bar{u}(j,\beta_j(q');\tau(q);q_i)$ by
\begin{align*}
    &\bar{u}(j,\beta_j(q_i);\tau(q);q_i)-\bar{u}(j,\beta_j(q');\tau(q);q_i)
    \\=&v(q_i)(x_j(\Phi_j(q_i))-x_j(\Phi_j(q')))-\beta_j(q_i)+\beta_j(q')
    \\=&v(q_i)(x_j(\Phi_j(q_i))-x_j(\Phi_j(q')))+\int_{t=q'}^{q_i}v(t)(-x_j'(\Phi_j(t)))d\Phi_j(t)
    \\=&\int_{t=q'}^{q_i}(v(t)-v(q_i))(-x_j'(\Phi_j(t)))d\Phi_j(t)
    \\ \geq&0
\end{align*}
The last inequality is obtained by discussing the cases of $q'\leq q$ and $q'\geq q$. In the former case, $v(t)\geq v(q_i)$ holds for all $t\in[q',q_i]$; In the latter case, $v(t)\leq v(q_i)$ holds for all $t\in[q_i,q']$. This completes the proof of the claim.
\end{proof}

By \Cref{claim.lem2.6}, it suffices to verify that for any $j\in [m]$ with $\Phi'_j(q_i)>0$, it holds that $(j,\beta_j(q_i))\in\arg\max_{j'\in[m]}$ $\bar{u}(j',\beta_{j'}(q_i);\tau(q);q_i)$. From the proof of \Cref{claim.lem2.6}, we know that $\bar{u}(j',\beta_{j'}(q_i);\tau(q);q_i)=v(q_i)x_{j'}(\Phi_{j'}(q_i))-\beta_{j'}(q_i)$. Thus, this is precisely condition 2.

\subsubsection{Necessity:} ~
Suppose $\vec{\Phi}(q)$ is the cummulation of $\vec{\pi}(q)$, which forms an sBNE with $\vec{\beta}(q)$. We may use $\Phi_j(q)$ and $\Phi_{\pi_j}(q)$ interchangeably. Next, we prove that the conditions 1 and 2 are satisfied.

We first demonstrate that condition 1 holds. According to \Cref{lemma:SBNE-deterministic-effort}, for each $j\in[m]$, there exists a continuous and strictly decreasing function $\hat{\beta}_j(\phi)$ such that $\beta_j(q_i)=\hat{\beta}_j(\Phi_{\pi_j}(q_i))$ for all $q_i\in[0,1]$ with $\pi_j(q_i)>0$. For convenience, for any $q_i$ with $\pi_j(q_i)=0$, we also define $\beta_j(q_i)=\hat{\beta}_j(\Phi_{\pi_j}(q))$. It follows that for any $q_i\in[0,1]$, $\Psi_j^{\tau(q)}(\beta_j(q_i))=\Phi_{\pi_j}(q_i)$. Now for any $q_i\in[0,1]$ with $\pi_j(q_i)>0$, for any other $q'\in[0,1]$, we have 

$$v(q_i)x_j(\Phi_j(q_i))-\beta_j(q_i)\geq v(q_i)x_j(\Phi_j(q'))-\beta_j(q'),$$ 
i.e., 
\begin{align*}
    v(q_i)(x_j(\Phi_j(q_i))-x_j(\Phi_j(q')))\geq &\beta_j(q_i)-\beta_j(q')\\
    =&\hat{\beta}_j(\Phi_j(q_i))-\hat{\beta}_j(\Phi_j(q')).
\end{align*}
Based on above inequality, for any $q'>q$ with $\Phi_j(q')>\Phi_j(q_i)$, we get 

$$\frac{v(q_i)(x_j(\Phi_j(q'))-x_j(\Phi_j(q_i)))}{\Phi_j(q')-\Phi_j(q_i)}\leq\frac{\hat{\beta}_j(\Phi_j(q'))-\hat{\beta}_j(\Phi_j(q_i))}{\Phi_j(q')-\Phi_j(q_i)}.$$
When $\Phi_j(q')$ tends to $\Phi_j(q_i)+0$, we obtain that $\hat{\beta}_j'(\Phi_j(q_i))$ $\geq v(q_i)x_j'(\Phi_j(q_i))$.

Similarly, for any $q'$ with $\Phi_j(q')<\Phi_j(q_i)$, we have $$\frac{v(q_i)(x_j(\Phi_j(q_i))-x_j(\Phi_j(q')))}{\Phi_j(q_i)-\Phi_j(q')}\geq\frac{\hat{\beta}_j(\Phi_j(q_i))-\hat{\beta}_j(\Phi_j(q'))}{\Phi_j(q_i)-\Phi_j(q')}.$$
When $\Phi_j(q')$ tends to $\Phi_j(q_i)-0$, we know that  $\hat{\beta}_j'(\Phi_j(q_i))$ $\leq v(q_i)x_j'(\Phi_j(q_i))$.

Thus, we obtain $\hat{\beta}_j'(\Phi_j(q_i))=v(q_i)x_j'(\Phi_j(q_i))$ for any $q_i\in[0,1]$ with $\pi_j(q_i)>0$. It follows that for all $q_i\in[0,1]$, $\hat{\beta}_j(\Phi_j(1))-\hat{\beta}_j(\Phi_j(q_i))=\int_{t=q_i}^{1}v(q_i)x_j'(\Phi_j(t))\pi_j(t)dt$. It is hot hard to see that $\beta_j(1)=0$, otherwise exerting $0$ effort will be strictly better. Therefore, for any $q_i\in[0,1]$ we have 
\begin{align*}
    \beta_j(q_i)=&\hat{\beta}_j(\Phi_j(q_i))\\
    =&\beta_j(1)-\int_{t=q_i}^{1}v(q_i)x_j'(\Phi_j(t))\pi_j(t)dt\\
    =&\int_{t=q_i}^{1}v(q_i)(-x_j'(\Phi_j(t)))\pi_j(t)dt\\
    =&\int_{t=q_i}^{1}v(q_i)(-x_j'(\Phi_j(t)))\Phi_j'(t)dt.
\end{align*}
Note that the last equation holds because $\Phi_j'(q)=\pi_j(q)$ holds almost everywhere for $q\in[0,1]$.

Next, we prove that condition 2 holds. For any $j\in[m]$ and $q_i\in[0,1]$ such that $\Phi_j'(q_i)>0$, we show that $v(q_i)x_j(\Phi_j(q_i))-\beta_j(q_i)=\max_{j'\in[m]}(v(q_i)x_{j'}(\Phi_{j'}(q_i))-\beta_{j'}(q_i))$. Since $\vec{\pi}(q)$ and $\vec{\beta}(q)$ form an sBNE, this holds for all $q_i\in[0,1]$ that $\pi_j(q_i)>0$. 

For any $q_i\in[0,1]$ with $\Phi_j'(q_i)>0$, $q_i$ is a cluster point of $\{q'\in[0,1]:\pi_j(q')>0\}$, otherwise there exists some $q'_1<q_i<q'_2$ such that $\Phi_j(q)$ is constant on $(q'_1,q'_2)$, contradicting with $\Phi_j'(q_i)>0$. Let such $q'$ tend to $q_i$, by definition of sBNE, we know $v(q')x_j(\Phi_j(q'))-\beta_j(q')-\max_{j'\in[m]}(v(q')x_{j'}(\Phi_{j'}(q'))-\beta_{j'}(q'))=0$ for all $q'$. By continuity, we finally get the condition 2, $v(q_i)x_j(\Phi_j(q_i))-\beta_j(q_i)-\max_{j'\in[m]}(v(q_i)x_{j'}(\Phi_{j'}(q_i))-\beta_{j'}(q_i))=0$. 
\end{proof}

\subsection{Proof of Theorem \ref{thm:contestant-equilibrium-choice-condition}}

\begin{proof}
Denote the expected utility as\begin{align*}
    \hat{u}_j(q_i)=&v(q_i)x_j(\Phi_j(q_i))-\beta_j(q_i)\\
    =&v(q_i)x_j(\Phi_j(q_i))-\int_{q_i}^1 v(t)\left(-x_j'(\Phi_j(t))\right)\Phi_j'(t)dt.
\end{align*} 
Recall the condition 2 in \Cref{lemma:contestant-equilibrium-effort-condition} which states that for all $q_i\in[0,1]$
\begin{equation}
\label{eq:argmax-uj}
    \{j\in[m]:\Phi_j'(q_i)>0\}\subseteq \arg\max_{j\in[m]} \hat{u}_j(q_i).
\end{equation}
% By condition 1 in \Cref{lemma:contestant-equilibrium-effort-condition}, we have 
% \begin{align*}
%     \hat{u}_j(q_i)&=v(q_i)x_j(\Phi_j(q_i))-\int_{q_i}^1 v(t)\left(-x_j'(\Phi_j(t))\right)\Phi_j'(t)dt\\
%     =&v(q_i)x_j(\Phi_j(q_i))+\int_{q_i}^1 v(t)\frac{dx_j(\Phi_j(t))}{dt}dt\\
%     =&v(q_i)x_j(\Phi_j(q_i))+v(1)x_j(\Phi_j(1))-v(q_i)x_j(\Phi_j(q_i))\\& -\int_{q_i}^1 v'(t)x_j(\Phi_j(t))dt\\
%     =&v(1)x_j(\Phi_j(1))-\int_{q_i}^1 v'(t)x_j(\Phi_j(t))dt.
% \end{align*}
% \lny{Here for convenience we assume $v(q)$ is differentiable.}

We prove that \Cref{eq:argmax-uj} holds for all $q_i\in[0,1]$ if and only if $\{j\in[m]:\Phi_j'(q_i)>0\}\subseteq \arg\max_{j\in[m]}x_j(\Phi_j(q_i))$ holds for all $q_i\in[0,1]$.

\subsubsection{Sufficiency:} ~
Assume that $\{j\in[m]:\Phi_j'(q_i)>0\}\subseteq \arg\max_{j\in[m]}$ $x_j(\Phi_j(q_i))$ holds for all $q_i\in[0,1]$. Define $X(q_i):=\max_{j\in[m]}x_j(\Phi_j(q_i))$. Then for any $j\in[m]$ and $q_i\in[0,1]$, we have $x_j(\Phi_j(q_i))\leq X(q_i)$, and the equality holds when $\Phi_j'(q_i)>0$. 

For any $j\in[m]$ and $q_i\in[0,1]$, applying integration by parts, we obtain
\begin{align*}
    \hat{u}_j(q_i)&=v(q_i)x_j(\Phi_j(q_i))-\int_{q_i}^1 v(t)\left(-x_j'(\Phi_j(t))\right)\Phi_j'(t)dt\\
    &=v(q_i)x_j(\Phi_j(q_i))+\int_{t=q_i}^1 v(t)dx_j(\Phi_j(t))\\
    &=v(q_i)x_j(\Phi_j(q_i))+v(1)x_j(\Phi_j(1))-v(q_i)x_j(\Phi_j(q_i))\\& -\int_{t=q_i}^1 x_j(\Phi_j(t))dv(t)\\
    &=v(1)x_j(\Phi_j(1))-\int_{t=q_i}^1 x_j(\Phi_j(t))dv(t).
\end{align*}
As $v(q)$ is decreasing and $x_j(\Phi_j(t))\leq X(t)$ for all $t\in[0,1]$, we have 

$$\hat{u}_j(q_i)\leq v(1)X(1)-\int_{t=q_i}^1 X(t)dv(t).$$

Consider any $j\in[m]$. We will show that, for all $q_i\in[0,1]$ with $\Phi'_j(q_i)>0$, it holds that $\hat{u}_j(q_i)=v(1)X(1)-\int_{t=q_i}^1 X(t)dv(t)$, which implies that $j\in \arg\max_{j\in[m]}\hat{u}_j(q_i)$. Note that if $\Phi_j(1)=0$, then it holds that $\Phi'_j(q_i)=0$ for all $q_i\in[0,1]$, so the claim holds trivially.
Thus, we consider the case of  $\Phi_j(1)>0$. 

Let $\underline{q}_j=\inf\{q_i\in[0,1]:\Phi'_j(q_i)>0\}$. We claim that for all $q_i\in[\underline{q}_j,1]$, we have $x_j(\Phi_j(q_i))=X(q_i)$. If this claim holds, we can substitute $x_j(\Phi_j(q_i))$ with $X(q_i)$ in $\hat{u}_j(q_i)$ and demonstrate that $\hat{u}_j(q_i)$ achieves its maximum value.

Now, we prove this claim. For any $q_i\in[\underline{q}_j,1]$ with $\Phi_j'(q_i)>0$,  the claim holds by assumption. Therefore, we only need to prove it for $q_i\in(\underline{q}_j,1]$ such that either $\Phi_j'(q_i)=0$ or $\Phi_j$ is not differentiable at $q_i$. 
Take $\underline{q}'=\sup\{q'\in[\underline{q}_j,q_i]:\Phi'_j(q')>0\}$. We have $\underline{q}'\leq q_i$ and $\Phi_j(q_i)=\Phi_j(\underline{q}')+\int_{\underline{q}'}^{q_i}\Phi'_j(t)dt=\Phi_j(\underline{q}')$. Since it holds that $x_j(\Phi_j(q'))= X(q')$ for all $q'$ with $\Phi_j'(q')>0$, by continuity of $x_j(q)$ and $X(q)$, $x_j(\Phi_j(\underline{q}'))=X(\underline{q}')$ also holds. As $X(q)$ is non-increasing, we have $X(q_i)\leq X(\underline{q}')$. Combining these together, we have 

$$X(q_i)\leq X(\underline{q}')=x_j(\Phi_j(\underline{q}'))=x_j(\Phi_j(q_i)).$$
On the other hand, by definition, we know $X(q_i)\geq x_j(\Phi_j(q_i))$. Thus, $x_j(\Phi_j(q_i))=X(q_i)$.

For any $q_i\in[0,1]$ with $\Phi_j'(q_i)>0$, we have $[q_i,1]\subseteq [\underline{q}_j,1]$, and by the above claim we obtain \begin{align*}
    \hat{u}_j(q_i)=v(1)X(1)-\int_{q_i}^1 X(t)dv(t).
\end{align*}
This completes the proof of sufficiency.

\subsubsection{Necessity:} ~
Assume that \Cref{eq:argmax-uj} holds for all $q_i\in[0,1]$. We prove that for all $j\in[m]$ and $q_i\in[0,1]$ with $\Phi_j'(q_i)>0$, it holds that $x_j(\Phi_j(q_i))=\max_{j'\in[m]}x_{j'}(\Phi_{j'}(q_i))$. 

Suppose, for the sake of contradiction, that there exists $j\in[m]$ and $q_i\in[0,1]$, such that $\Phi_j'(q_i)>0$ but $x_j(\Phi_j(q_i))<\max_{j'\in[m]}x_{j'}(\Phi_{j'}(q_i))$. Let $j'$ be a contest satisfying that $x_j(\Phi_j(q_i))<x_{j'}(\Phi_{j'}(q_i))$. We discuss three cases:

\noindent \textbf{(a)} $\hat{u}_j(q_i)\leq \hat{u}_{j'}(q_i)$.

First, since $x_j(\Phi_j(q_i))<x_{j'}(\Phi_{j'}(q_i))$, by the continuity of $x_j$ and $x_{j'}$, there exists $q_1<q_i$ such that $x_j(\Phi_j(q))<x_{j'}(\Phi_{j'}(q))$ for all $q\in[q_1,q_i]$. Due to $\Phi_j'(q_i)>0$, we have $\Phi_j(q_1)<\Phi(q_i)$, and therefore there exists $q_2\in(q_1,q_i)$ such that $\Phi_j'(q_2)>0$. By definition, we have 

$$\hat{u}_j(q_2)=\hat{u}_j(q_i)-\int_{t=q_2}^{q_i}x_j(\Phi_j(t))dv(t)$$ 
and 
$$\hat{u}_{j'}(q_2)=\hat{u}_{j'}(q_i)-\int_{t=q_2}^{q_i}x_{j'}(\Phi_{j'}(t))dv(t).$$ 
Since $v(q)$ is decreasing, it follows that $\hat{u}_j(q_2)<\hat{u}_{j'}(q_2)$, but $\Phi_j'(q_2)>0$, contradicting with \Cref{eq:argmax-uj}.

\noindent \textbf{(b)} $\hat{u}_j(q_i)>\hat{u}_{j'}(q_i)$ and $\Phi_{j'}(q_i)=\Phi_{j'}(1)$. 

In this case, for any $q'\in[q_i,1]$, we have $\Phi_{j'}(q')=\Phi_{j'}(q_i)$. Therefore, it holds that  

$$x_j(\Phi_j(q'))\leq x_j(\Phi_j(q_i))<x_{j'}(\Phi_{j'}(q_i))=x_{j'}(\Phi_{j'}(q')).$$
It follows that 
\begin{align*}
    \hat{u}_j(q_i)&=v(1)x_j(\Phi_j(1))-\int_{t=q_i}^1 x_j(\Phi_j(t))dv(t)\\
    &<v(1)x_{j'}(\Phi_{j'}(1))-\int_{t=q_i}^1 x_{j'}(\Phi_{j'}(t))dv(t)\\
    &=\hat{u}_{j'}(q_i),
\end{align*}
which contradicts with \Cref{eq:argmax-uj}.

\noindent \textbf{(c)} $\hat{u}_j(q_i)>\hat{u}_{j'}(q_i)$ and $\Phi_{j'}(q_i)<\Phi_{j'}(1)$. 

Let $q_1=\inf\{q'\in[q_i,1]:\Phi_{j'}'(q')>0\}$, then $\Phi_{j'}(q_1)=\Phi_{j'}(q_i)$. By \Cref{eq:argmax-uj}, we know that $\hat{u}_{j'}(q')\geq \hat{u}_{j}(q')$ for all $q'$ with $\Phi_{j'}'(q')>0$. Thus, as $q'$ tends to $q_1$, by continuity we have $\hat{u}_{j'}(q_1)\geq \hat{u}_{j}(q_1)$. %\lny{check}. 
Similar to case (b), for any $q'\in[q_i,q_1]$, we have $\Phi_{j'}(q')=\Phi_{j'}(q_i)$, and therefore $x_j(\Phi_j(q'))<x_{j'}(\Phi_{j'}(q'))$.
It follows that 
\begin{align*}
    \hat{u}_j(q_i)&=\hat{u}_j(q_1)-\int_{t=q_i}^{q_1} x_j(\Phi_j(t))dv(t)\\
    &<\hat{u}_{j'}(q_1)-\int_{t=q_i}^{q_1} x_{j'}(\Phi_{j'}(t))dv(t)\\
    &=\hat{u}_{j'}(q_i),
\end{align*}
which contradicts.

This completes the proof of necessity.
\end{proof}

\subsection{Proof of \Cref{theorem:contestant-equilibrium-wine22}}
\begin{proof}
    Combining \Cref{thm:contestant-equilibrium-choice-condition} and applying the characterization on cummulative equilibrium behavior in \cite{DLLQ22}, we can get the cumulative equilibrium choice strategy in this model. 
\end{proof}

\subsection{Proof of Lemma \ref{lemma:known-competitor-number-same-prize}}
\begin{proof}
Define $p=\Phi_j(1)$ and $q=\Phi_j(q_i)$. We will prove that
\begin{align*}
    &\sum_{k=1}^n\binom{n-1}{k-1}p^{k-1}(1-p)^{n-k}\sum_{l=1}^{k}w_{j,l}\binom{k-1}{l-1}(\frac{q}{p})^{l-1}(1-\frac{q}{p})^{k-l}
    \\=&\sum_{k=1}^nw_{j,k}\binom{n-1}{k-1}q^{k-1}(1-q)^{n-k}.
\end{align*}
For convenience, we rewrite $n-1$ as $n$, $k-1$ as $k$, and $l-1$ as $l$, then we need to prove that
\begin{align*}
    &\sum_{k=0}^{n}\binom{n}{k}p^{k}(1-p)^{n-k}\sum_{l=0}^{k}w_{j,l+1}\binom{k}{l}(\frac{q}{p})^{l}(1-\frac{q}{p})^{k-l}
    \\=&\sum_{k=0}^nw_{j,k+1}\binom{n}{k}q^{k}(1-q)^{n-k}.
\end{align*}
We show that the both sides have equal coefficient for each $w_{j,l+1}$, as follows: 
\begin{align*}
&\binom{n}{l}q^{l}(1-q)^{n-l}
\\=&\binom{n}{l}p^{l}(\frac{q}{p})^{l}(p-q+1-p)^{n-l}
\\=&\binom{n}{l}p^{l}(\frac{q}{p})^{l}\sum_{t=0}^{n-l}\binom{n-l}{t}(p-q)^t(1-p)^{n-l-t}
\\=&\binom{n}{l}p^{l}(\frac{q}{p})^{l}\sum_{k=l}^{n}\binom{n-l}{k-l}(p-q)^{k-l}(1-p)^{n-k}\\
=&\sum_{k=l}^{n}\binom{n}{n-l}\binom{n-l}{n-k}p^{k}(1-\frac{q}{p})^{k-l}(1-p)^{n-k}(\frac{q}{p})^{l}
\\=&\sum_{k=l}^{n}\binom{n}{n-k}\binom{k}{l}p^{k}(1-p)^{n-k}(\frac{q}{p})^{l}(1-\frac{q}{p})^{k-l}.
\end{align*}
This completes the proof. 
\end{proof}

\section{Missing Proofs in Section \ref{sec:designer}}

\subsection{Proof of Lemma \ref{lemma:designer-utility-effort-objective}}

\begin{proof}
Let $\beta^*_j(q)$ denote the effort strategy in contest $j$ induced by $\Phi_j^*(q)$, as described in \Cref{lemma:contestant-equilibrium-effort-condition}. For each $k\in[n]$, the expected effort of the $k$-th contestant in contest $j$ is given by 
\begin{align*}
    \E[e_j^{(k)}]&=n\E[e_i|J_i=j\land \mathrm{rank}(i,\boldsymbol{J},\boldsymbol{e})=k]
    \\&=n\int_{0}^1\beta_j^*(t)\Pr[\mathrm{rank}(i,\boldsymbol{J},\boldsymbol{e})=k|q_i=t]{\Phi_j^*}'(t)dt
    \\&=n\int_{0}^1\beta_j^*(t)g_k(\Phi_j^*(t)){\Phi_j^*}'(t)dt.
\end{align*}
Recall that $\beta_j^*(q_i)=\int_{q_i}^1v(t)(-\frac{dx_j(\Phi_j^*(t))}{dt})dt$. Observe that $x_j(\Phi_j^*(t))=x_j(x_j^{-1}(Q^{-1}(t)))=Q^{-1}(t)$ whenever $\Phi_j^*(t)>0$. Thus, we have $\beta_j^*(q_i)=\int_{t=q_i}^1v(t)d(-Q^{-1}(t))$.

% Define $g_k(\phi):=\Pr[\mathrm{rank}(i,\boldsymbol{J},\boldsymbol{e})=k|\Phi_j^*(q_i)=\phi]=\binom{n-1}{k-1}\phi^{k-1}(1-\phi)^{n-k}$
It follows that
\begin{align*}
    \E[e_j^{(k)}]
    =&n\int_{t=0}^1\int_{s=t}^1v(s)d(-Q^{-1}(s))\cdot g_k(\Phi_j^*(t)){\Phi_j^*}'(t)dt
    \\=&n\int_{s=0}^1v(s)\int_{0}^s g_k(\Phi_j^*(t)){\Phi_j^*}'(t)dtd(-Q^{-1}(s))
    \\=&n\int_{s=0}^1v(s)\int_0^{\Phi_j^*(s)} g_k(z)dzd(-Q^{-1}(s))
\end{align*}
By definition, $G(\phi;\vec{\alpha}_j)=n\sum_{k=1}^n\alpha_{j,k}\int_0^{\phi} g_k(t)dt$, the expected utility of designer $j$ can be calculated as follows:
\begin{align*}
&\hat{R}_j(x_j(\phi),x_{-j}^{-1}(\phi),\vec{\alpha}_j)
\\=&\sum_{k=1}^n\alpha_{j,k}\E[e_j^{(k)}]
\\=&\int_{s=0}^1v(s)G(\Phi_j^*(s);\vec{\alpha}_j)d(-Q^{-1}(s))
\\=&\int_{0}^{+\infty}v(Q(x))G(x_j^{-1}(x);\vec{\alpha}_j)dx
\\=&\int_{0}^{+\infty}v(x_j^{-1}(x)+x_{-j}^{-1}(x))G(x_j^{-1}(x);\vec{\alpha}_j)dx.
\end{align*}
Note that if we define $v(q)=0$ for all $q>1$, the third equation holds by changing variable as $x=Q^{-1}(s)$, for which $\Phi_j^*(s)=x_j^{-1}(Q^{-1}(t))=x_j^{-1}(x)$. 
\end{proof}

\subsection{Proof of Lemma \ref{lemma: sufficient condition on weight-monotone}}
\begin{proof}
    Observe that 
    $$\sum_{k=1}^n\alpha_{j,k}g_k(\phi)=\sum_{k=1}^{n}k(\alpha_{j,k}-\alpha_{j,k+1})\xi_k(\phi),$$
    where $\alpha_{j,n+1}$ is defined as $0$. Since $\xi_k(\phi)$ is non-increasing on $[0,1]$ and $(\alpha_{j,k}-\alpha_{j,k+1})$ is non-negative for all $k\in[n]$, it follows that $\sum_{k=1}^n\alpha_{j,k}g_k(\phi)$ is non-increasing in $\phi\in[0,1]$. Thus, $\vec{\alpha}_j$ is weight-monotone by definition.
\end{proof}

\subsection{Proof of Theorem \ref{thm: compare under single-crossing}}
\begin{proof}
By \Cref{lemma:designer-utility-effort-objective}, the objective function can be expressed as:
\begin{align*}
    &\hat{R}_j(x_{\vec{w}_j}(\phi),x_{-j}^{-1}(x),\vec{\alpha}_j)\\
=&\int_{0}^{+\infty}v(x_{\vec{w}_j}^{-1}(x)+x_{-j}^{-1}(x))G(x_{\vec{w}_j}^{-1}(x);\vec{\alpha}_j)dx.
\end{align*}
By changing variable as $x=x_{\vec{w}_j}(\phi)$, we have 
\begin{align}
    &\hat{R}_j(x_{\vec{w}_j}(\phi),x_{-j}^{-1}(x),\vec{\alpha}_j)\nonumber\\
=&\int_{0}^{1}v(\phi+x_{-j}^{-1}(x_{\vec{w}_j}(\phi)))G(\phi;\vec{\alpha}_j)\frac{-d x_{\vec{w}_j}(\phi)}{d\phi}d\phi.\label{eq:thm2-Rj}
\end{align}

We first prove the following two claims:
\begin{claim}\label{claim.thm2.1}
    $\int_0^1\phi\frac{-d x_{\vec{w}_j}(\phi)}{d\phi}d\phi\geq \int_0^1\phi\frac{-d x_{\vec{w}_j'}(\phi)}{d\phi}d\phi$.
\end{claim}
\begin{proof}
    We have the following through integration by parts:
\begin{align*}
    \int_0^1\phi\frac{-d x_{\vec{w}_j}(\phi)}{d\phi} d\phi
    =&-\phi x_{\vec{w}_j}(\phi)\vert_0^1+\int_0^1x_{\vec{w}_j}(\phi)d\phi\\
    =&-x_{\vec{w}_j}(1)+\int_0^1x_{\vec{w}_j}(\phi)d\phi
\end{align*}
Similarly, we have \begin{align*}
    \int_0^1\phi\frac{-d x_{\vec{w}_j'}(\phi)}{d\phi} d\phi
    =-x_{\vec{w}_j'}(1)+\int_0^1x_{\vec{w}_j'}(\phi)d\phi.
\end{align*}
Recall that  $\int_0^{1}x_{\vec{w}_j}(\phi)d\phi \geq\int_0^{1}x_{\vec{w}_j'}(\phi)d\phi$ and $x_{\vec{w}_j}(1)\leq x_{\vec{w}_j'}(1)$ by the definition of single-crossing-dominating. The claim follows immediately.
\end{proof}
\begin{claim}\label{claim.thm2.2}
    $\frac{G(\phi;\vec{\alpha}_j)}{\phi}$ is non-increasing in $\phi\in(0,1]$.
\end{claim}
\begin{proof}
    Since $\vec{\alpha}_j$ is weight-monotone and non-negative, we have that $G(\phi;\vec{\alpha}_j)$ is concave and non-decreasing in $\phi$. For any $0<\phi_1<\phi_2$, we have $G(\phi_1;\vec{\alpha}_j)\geq \frac{\phi_2-\phi_1}{\phi_2}G(0;\vec{\alpha}_j)+\frac{\phi_1}{\phi_2}G(\phi_2;\vec{\alpha}_j)$. By definition $G(\phi_1;\vec{\alpha}_j)=0$, so it follows that $G(\phi_1;\vec{\alpha}_j)\geq \frac{\phi_1}{\phi_2}G(\phi_2;\vec{\alpha}_j)$, i.e., $\frac{G(\phi_1;\vec{\alpha}_j)}{\phi_1}\geq \frac{G(\phi_2;\vec{\alpha}_j)}{\phi_2}$.
\end{proof}

To help our analysis, we define the following two functions: \begin{align*}
    &V(\phi,x):=v(\phi+x_{-j}^{-1}(x))\frac{G(\phi;\vec{\alpha}_j)}{\phi},\\
    &S_{\vec{w}_j}(t):=\int_0^{t}\phi\frac{-d x_{\vec{w}_j}(\phi)}{d\phi}d\phi.
\end{align*}
We can rewrite \Cref{eq:thm2-Rj} as 
\begin{align*}
    \hat{R}_j(x_{\vec{w}_j}(\phi),x_{-j}^{-1}(x),\vec{\alpha}_j)
=&\int_{0}^{1}V(\phi,x_{\vec{w}_j}(\phi))\frac{d S_{\vec{w}_j}(\phi)}{d\phi}d\phi\\
=&\int_{\phi=0}^{1}V(\phi,x_{\vec{w}_j}(\phi))d S_{\vec{w}_j}(\phi).
\end{align*}
And similarly, 
\begin{align*}
    \hat{R}_j(x_{\vec{w}_j'}(\phi),x_{-j}^{-1}(x),\vec{\alpha}_j)=\int_{\phi=0}^{1}V(\phi,x_{\vec{w}_j'}(\phi))d S_{\vec{w}_j'}(\phi).
\end{align*}
To show the theorem, we need to prove that 

$$\int_{\phi=0}^{1}V(\phi,x_{\vec{w}_j}(\phi))d S_{\vec{w}_j}(\phi)\geq \int_{\phi=0}^{1}V(\phi,x_{\vec{w}_j'}(\phi))d S_{\vec{w}_j'}(\phi).$$
By \Cref{claim.thm2.1}, we have $S_{\vec{w}_j}(1)\geq S_{\vec{w}_j'}(1)$. Since $V(\phi,x)$ is always non-negative, the following claim is sufficient to prove the theorem.
\begin{claim}\label{claim.thm2.3}
    For any $\phi_1,\phi_2$ such that $S_{\vec{w}_j}(\phi_1)=S_{\vec{w}_j'}(\phi_2)$, it holds that $V(\phi_1,x_{\vec{w}_j}(\phi_1))\geq V(\phi_2,x_{\vec{w}_j'}(\phi_2))$.
\end{claim}
Before proving \Cref{claim.thm2.3}, we prove the following two claims using the single-crossing-dominating condition between $x_{\vec{w}_j}(\phi)$ and $x_{\vec{w}_j'}(\phi)$.
\begin{claim}\label{claim.thm2.4}
    For any $\phi\in[0,1]$, $S_{\vec{w}_j}(\phi)\geq S_{\vec{w}_j'}(\phi)$. 
\end{claim}
\begin{proof}
    Since $\frac{-dx_{\vec{w}_j}(\phi)}{d\phi}$ is single-crossing with respect to $\frac{-dx_{\vec{w}_j'}(\phi)}{d\phi}$, by definition there exists $\hat{\phi_0}\in[0,1]$ such that $\frac{-dx_{\vec{w}_j}(\phi)}{d\phi}\geq \frac{-dx_{\vec{w}_j'}(\phi)}{d\phi}$ for all $\phi\in[0,\hat{\phi_0}]$ and $\frac{-dx_{\vec{w}_j}(\phi)}{d\phi}\leq \frac{-dx_{\vec{w}_j'}(\phi)}{d\phi}$ for all $\phi\in[\hat{\phi_0},1]$. 
    % One can see that $\hat{\phi_0}\geq \phi_0$ because $\frac{dx_{\vec{w}_j}(\phi)}{d\phi}\big|_{\phi=\phi_0}\leq \frac{dx_{\vec{w}_j'}(\phi)}{d\phi}\big|_{\phi=\phi_0}$ by the definition of single crossing.
    Then for any $\phi\in[0,\hat{\phi_0}]$, we have $S_{\vec{w}_j}(\phi)=\int_0^{\phi}\phi\frac{-dx_{\vec{w}_j}(\phi)}{d\phi}d\phi\geq\int_0^{\phi}\phi\frac{-dx_{\vec{w}_j'}(\phi)}{d\phi}d\phi =S_{\vec{w}_j'}(\phi)$. For any $\phi\in[\hat{\phi_0},1]$, we have 
    \begin{align*}
        S_{\vec{w}_j}(\phi)
        =&S_{\vec{w}_j}(1)-\int_{\phi}^{1}\phi\frac{-dx_{\vec{w}_j}(\phi)}{d\phi}d\phi\\
        \geq&S_{\vec{w}_j'}(1)-\int_{\phi}^{1}\phi\frac{-dx_{\vec{w}_j'}(\phi)}{d\phi}d\phi\\
        =& S_{\vec{w}_j'}(\phi). \qedhere
    \end{align*}
\end{proof}
\begin{claim}\label{claim.thm2.5}
    For any $X\in[0,+\infty)$, $\int_{X}^{+\infty}x_{\vec{w}_j}^{-1}(x)dx\geq \int_{X}^{+\infty}x_{\vec{w}_j'}^{-1}(x)dx$.
\end{claim}
\begin{proof}
Since $x_{\vec{w}_j}(\phi)$ is single-crossing with respect to $x_{\vec{w}_j'}(\phi)$, there exists $\phi_0\in[0,1]$ such that $x_{\vec{w}_j}(\phi)\geq x_{\vec{w}_j'}(\phi)$ for all $\phi\in[0,\phi_0]$ and $x_{\vec{w}_j}(\phi)\leq x_{\vec{w}_j'}(\phi)$ for all $\phi\in[\phi_0,1]$. 
Let $X_0=x_{\vec{w}_j}(\phi_0)$, then by the non-increase of $x_{\vec{w}_j}^{-1}(x)$ and $x_{\vec{w}_j'}^{-1}(x)$ and the single-crossing property, we have that $x_{\vec{w}_j}^{-1}(X)\leq x_{\vec{w}_j'}^{-1}(X)$ for all $X\leq X_0$ and $x_{\vec{w}_j}^{-1}(X)\geq x_{\vec{w}_j'}^{-1}(X)$ for all $X\geq X_0$. For any $X\geq X_0$, we immediately obtain $\int_{X}^{+\infty}x_{\vec{w}_j}^{-1}(x)dx\geq \int_{X}^{+\infty}x_{\vec{w}_j'}^{-1}(x)dx$. For any $X\leq X_0$, we have 
    \begin{align*}
        \int_{X}^{+\infty}x_{\vec{w}_j}^{-1}(x)dx
        =&\int_{0}^{+\infty}x_{\vec{w}_j}^{-1}(x)dx-\int_{0}^{X}x_{\vec{w}_j}^{-1}(x)dx\\
        \geq&\int_{0}^{+\infty}x_{\vec{w}_j'}^{-1}(x)dx-\int_{0}^{X}x_{\vec{w}_j'}^{-1}(x)dx\\
        =& \int_{X}^{+\infty}x_{\vec{w}_j'}^{-1}(x)dx,
    \end{align*}
    where the inequality holds by 
    \begin{align*}
        \int_{0}^{+\infty}x_{\vec{w}_j}^{-1}(x)dx=&\int_{0}^{1}x_{\vec{w}_j}(\phi)d\phi\\
        \geq &\int_{0}^{1}x_{\vec{w}_j'}(\phi)d\phi=\int_{0}^{+\infty}x_{\vec{w}_j'}^{-1}(x)dx,
    \end{align*}   
    and $x_{\vec{w}_j}^{-1}(x)\leq x_{\vec{w}_j'}^{-1}(x)$ for all $x\leq X_0$.    
\end{proof}
Finally, we prove \Cref{claim.thm2.3} in the following.
\begin{proof}[Proof of \Cref{claim.thm2.3}]
    By the non-increase of $v(\phi)$ and $x_{-j}^{-1}(x)$, and \Cref{claim.thm2.2}, one can see that $V(\phi,x)$ is non-increasing in $\phi$ and non-decreasing in $x$. So we only need to prove that $\phi_1\leq \phi_2$ and $x_{\vec{w}_j}(\phi_1)\geq x_{\vec{w}_j'}(\phi_2)$.

    We firstly prove that $\phi_1\leq\phi_2$. By \Cref{claim.thm2.4} we have $S_{\vec{w}_j}(\phi_2)\geq S_{\vec{w}_j'}(\phi_2)$. By assumption $S_{\vec{w}_j}(\phi_1)=S_{\vec{w}_j'}(\phi_2)$, it holds that $S_{\vec{w}_j}(\phi_1)\leq S_{\vec{w}_j}(\phi_2)$. This implies that $\phi_1\leq\phi_2$.
    
    Next we prove that $x_{\vec{w}_j}(\phi_1)\geq x_{\vec{w}_j'}(\phi_2)$. 
    % We show this by proving that for any $X\in[0,+\infty)$, $\int_{X}^{+\infty}x_{\vec{w}_j}^{-1}(x)dx\geq \int_{X}^{+\infty}x_{\vec{w}_j'}^{-1}(x)dx$. 
    Observe that we have 
    
    $$S_{\vec{w}_j}(\phi_1)=\int_0^{\phi_1}\phi\frac{-d x_{\vec{w}_j}(\phi)}{d\phi}d\phi=\int_{x_{\vec{w}_j}(\phi_1)}^{+\infty}x_{\vec{w}_j}^{-1}(x)dx,$$ 
    by changing variable as $x=x_{\vec{w}_j}(\phi)$, and similarly $S_{\vec{w}_j'}(\phi_2)= \int_{x_{\vec{w}_j'}(\phi_2)}^{+\infty}x_{\vec{w}_j'}^{-1}(x)dx$.
    Then, combining this with $S_{\vec{w}_j}(\phi_1)= S_{\vec{w}_j'}(\phi_2)$, we have 
    $\int_{x_{\vec{w}_j}(\phi_1)}^{+\infty}x_{\vec{w}_j}^{-1}(x)dx=\int_{x_{\vec{w}_j'}(\phi_2)}^{+\infty}x_{\vec{w}_j'}^{-1}(x)dx.$

    By \Cref{claim.thm2.4}, we have
    
    $$\int_{x_{\vec{w}_j}(\phi_1)}^{+\infty}x_{\vec{w}_j}^{-1}(x)dx\geq \int_{x_{\vec{w}_j}(\phi_1)}^{+\infty}x_{\vec{w}_j'}^{-1}(x)dx,$$ 
    and 
    
    $$\int_{x_{\vec{w}_j'}(\phi_2)}^{+\infty}x_{\vec{w}_j'}^{-1}(x)dx\geq \int_{x_{\vec{w}_j}(\phi_1)}^{+\infty}x_{\vec{w}_j'}^{-1}(x)dx.$$ 
    This implies that $x_{\vec{w}_j'}(\phi_2)\leq x_{\vec{w}_j}(\phi_1)$.

In conclusion, we have $\phi_1\leq \phi_2$ and $x_{\vec{w}_j}(\phi_1)\geq x_{\vec{w}_j'}(\phi_2)$, and therefore $V(\phi_1,x_{\vec{w}_j}(\phi_1))\geq V(\phi_2,x_{\vec{w}_j'}(\phi_2))$.
\end{proof}
By \Cref{claim.thm2.3}, we have $\int_{\phi=0}^{1}V(\phi,x_{\vec{w}_j}(\phi))d S_{\vec{w}_j}(\phi)\geq \int_{\phi=0}^{1}V(\phi,x_{\vec{w}_j'}(\phi))d S_{\vec{w}_j'}(\phi)$. This completes the proof.
\end{proof}

\subsection{Proof of Theorem \ref{thm: winner-take-all dominant} }
\begin{proof}
By \Cref{thm: compare under single-crossing}, we only need to prove that $x_{\vec{w}_j^*}(\phi)$ single-crossing-dominates $x_{\vec{w}_j'}(\phi)$. We check the conditions in the definition. 

The first condition holds as $\int_0^1x_{\vec{w}_j^*}(\phi)d\phi=\frac1nT_j\geq \frac1n\sum_{k=1}^nw_{j,k}'\geq \int_0^1x_{\vec{w}_j'}(\phi)d\phi$.

Next we check the second condition that $x_{\vec{w}_j^*}(\phi)$ is single-crossing with respect to $x_{\vec{w}_j'}(\phi)$.
Firstly, by definition we have $x_{\vec{w}_j^*}(\phi)=T_j\xi_1(\phi)=T_j(1-\phi)^{n-1}$, and \begin{align*}
    x_{\vec{w}_j'}(\phi)=&\sum_{k=1}^nw_{j,k}'\binom{n-1}{k-1}\phi^{k-1}(1-\phi)^{n-k}.
    % \\=&\sum_{k=1}^nk(w_{j,k}'-w_{j,k+1}')\xi_k(\phi).
\end{align*}
For any $\phi\in(0,1)$ satisfying $x_{\vec{w}_j^*}(\phi)>0$,  we can calculate
\begin{align*}
    \frac{x_{\vec{w}_j'}(\phi)}{x_{\vec{w}_j^*}(\phi)}=&\frac{1}{T_j}\sum_{k=1}^nw_{j,k}'\binom{n-1}{k-1}(\frac{\phi}{1-\phi})^{k-1}.
\end{align*}

We discuss two cases:

\noindent \textbf{(a)}  For all $k>1$, $w_{j,k}'=0$. In this case, since $w_{j,1}'\leq T_j=w_{j,1}^*$, we have $x_{\vec{w}_j^*}(\phi)\geq x_{\vec{w}_j'}(\phi)$ for all $\phi\in[0,1]$. Therefore $x_{\vec{w}_j^*}(\phi)$ is single-crossing w.r.t. $x_{\vec{w}_j'}(\phi)$ by taking $\phi_0=1$ in \Cref{def:single-crossing}.

\noindent \textbf{(b)} There exists some $k>1$, $w_{j,k}'>0$. In this case ${x_{\vec{w}_j'}(\phi)}/{x_{\vec{w}_j^*}(\phi)}$ is strictly increasing in $\phi\in(0,1)$. Observe that $x_{\vec{w}_j^*}(0)=T_j> w_{j,1}'=x_{\vec{w}_j'}(0)$ and $x_{\vec{w}_j^*}(1)=0\leq w_{j,n}'=x_{\vec{w}_j'}(\phi)$. Therefore, there exists a zero-point $\phi_0\in(0,1]$ such that $x_{\vec{w}_j^*}(\phi_0)=x_{\vec{w}_j^*}(\phi_0)$. Since ${x_{\vec{w}_j'}(\phi)}/{x_{\vec{w}_j^*}(\phi)}$ is strictly increasing, such zero-point $\phi_0$ is unique. Therefore, we show that $x_{\vec{w}_j^*}(\phi)\geq x_{\vec{w}_j'}(\phi)$ for all $\phi\in[0,\phi_0]$ and $x_{\vec{w}_j^*}(\phi)\leq x_{\vec{w}_j'}(\phi)$ for all $\phi\in[\phi_0,1]$. By definition $x_{\vec{w}_j^*}(\phi)$ is single-crossing w.r.t. $x_{\vec{w}_j'}(\phi)$.

Lastly, we check the third condition that $\frac{-dx_{\vec{w}_j^*}(\phi)}{d\phi}$ is single-crossing with respect to $\frac{-dx_{\vec{w}_j'}(\phi)}{d\phi}$.

With some calculation we have 

$$\frac{-dx_{\vec{w}_j^*}(\phi)}{d\phi}=T_j(n-1)(1-\phi)^{n-2},$$ and 
\begin{align*}
    &\frac{-dx_{\vec{w}_j'}(\phi)}{d\phi}\\
    % =&\sum_{k=1}^nw_{j,k}'\binom{n-1}{k-1}((n-k)\phi^{k-1}(1-\phi)^{n-k-1}-(k-1)\phi^{k-2}(1-\phi)^{n-k})\\
    % =&\sum_{k=1}^nw_{j,k}'((n-k)\binom{n-1}{k-1}-k\binom{n-1}{k})\phi^{k-1}(1-\phi)^{n-k-1}\\
    =&\sum_{k=1}^{n-1}(w_{j,k}'-w_{j,k+1}')(n-1)\binom{n-2}{k-1}\phi^{k-1}(1-\phi)^{n-k-1}.
\end{align*}
For any $\phi\in(0,1)$,  we can calculate
\begin{align*}
    \frac{\frac{-dx_{\vec{w}_j'}(\phi)}{d\phi}}{\frac{-dx_{\vec{w}_j^*}(\phi)}{d\phi}}=&\frac{1}{T_j}\sum_{k=1}^{n-1}(w_{j,k}'-w_{j,k+1}')\binom{n-2}{k-1}(\frac{\phi}{1-\phi})^{k-1}.
\end{align*}

Since $w_{j,k}'-w_{j,k+1}'\geq 0$ for all $k\in[n-1]$, %$\frac{\frac{-dx_{\vec{w}_j'}(\phi)}{d\phi}}{\frac{-dx_{\vec{w}_j^*}(\phi)}{d\phi}}$ 
it is non-decreasing in $\phi\in[0,1)$. At the point $\phi=0$, we have 
\begin{align*}
     \frac{-dx_{\vec{w}_j^*}(\phi)}{d\phi}|_{\phi=0}=&T_j(n-1)\\
    \geq & (w_{j,1}'-w_{j,2}')(n-1)=\frac{-dx_{\vec{w}_j'}(\phi)}{d\phi}|_{\phi=0}.
\end{align*}
At the point $\phi=1$, we have 
$\frac{-dx_{\vec{w}_j^*}(\phi)}{d\phi}|_{\phi=1}=0 \leq \frac{-dx_{\vec{w}_j'}(\phi)}{d\phi}|_{\phi=1}.$
 Therefore, $\frac{-dx_{\vec{w}_j^*}(\phi)}{d\phi}$ is single-crossing with respect to $\frac{-dx_{\vec{w}_j'}(\phi)}{d\phi}$.
\end{proof}

\subsection{Proof of \Cref{cor:effort-objective-designer-SPE}}
\begin{proof}
    By \Cref{thm: winner-take-all dominant}, we know that the winner-take-all prize structure with all budget is a dominant strategy for each $j\in [m]$, when each $\vec{\alpha}_j$ is weight-monotone, which consequently form the designer equilibrium. Combining with the corresponding contestant equilibrium, we obtain an SPE.
\end{proof}

\subsection{Proof of Theorem \ref{thm: participant objective- simple contest}}
\begin{proof}
Let $\phi^*=\frac1n\max_{\vec{w}_j}\hat{R}_j(x_{\vec{w}_{j}}(\phi),x_{-j}^{-1}(x),\theta_j)$. Since $\phi^*\in[0,1]$, if $\phi^*=0$, then the theorem immediately holds, because any prize structure is optimal for designer $j$. 

Now we consider $\phi^*>0$. For convenience, let’s assume the optimal prize structure satisfies $w_{j,1}>w_{j,n}$, and we will address the case $w_{j,1}=w_{j,n}$ at the end of this proof. We first prove the following claim:
\begin{claim}\label{claim.thm4.7}
     For any strictly decreasing interim allocation functions $x_j(\phi)$ and $\tilde{x}_j(\phi)$, let $\phi_0=\frac1n\hat{R}_j(x_{j}(\phi),x_{-j}^{-1}(x)$ $,\theta_j)=\Phi_j^*(\theta_j;x_{j}(\phi),x_{-j}^{-1}(x))$. If $\phi_0>0$, then the following statements hold:
     
    If $x_{j}(\phi_0)\geq\tilde{x}_{j}(\phi_0)$, then 
    $$\hat{R}_j(x_{j}(\phi),x_{-j}^{-1}(x),\theta_j)\geq \hat{R}_j(\tilde{x}_{j}(\phi),x_{-j}^{-1}(x),\theta_j);$$
    
    If $x_{j}(\phi_0)\leq\tilde{x}_{j}(\phi_0)$, then 
    $$\hat{R}_j(x_{j}(\phi),x_{-j}^{-1}(x),\theta_j)\leq \hat{R}_j(\tilde{x}_{j}(\phi),x_{-j}^{-1}(x),\theta_j).$$
\end{claim}

\begin{proof}
    Define $Q_0(x)=x_{j}^{-1}(x)+x_{-j}^{-1}(x)$ and $Q_1(x)=\tilde{x}_{j}^{-1}(x)+x_{-j}^{-1}(x)$. By definition, $\phi_0=x_{j}^{-1}(Q_0^{-1}(\theta_j))$. As $\phi_0>0$, we have $x_j(\phi_0)=Q_0^{-1}(\theta_j)$. 
We also define $\phi_1= \hat{R}_j(\tilde{x}_{j}(\phi),x_{-j}^{-1}(x),\theta_j)=\tilde{x}_{j}^{-1}(Q_1^{-1}(\theta_j))$.

For the first statement, assuming $x_{j}(\phi_0)\geq\tilde{x}_{j}(\phi_0)$, we prove $\phi_1\leq\phi_0$ by contradiction. Suppose $\phi_1>\phi_0$, then $\tilde{x}_{j}(\phi_1)<\tilde{x}_{j}(\phi_0)\leq x_j(\phi_0)$. Denote $X_0=x_j(\phi_0)$ and $X_1=x_j(\phi_1)$. As $\phi_1>\phi_0>0$, we have $Q_0^{-1}(\theta_j)=x_j(\phi_0)=X_0$ and $Q_1^{-1}(\theta_j)=x_j(\phi_1)=X_1$. We discuss on the continuity of $x_{-j}^{-1}$ at $X_0$:

If $x_{-j}^{-1}$ is discontinuous at $X_0$, then there exists some contest $j'\neq j$ that $w_{j',1}=\cdots=w_{j',n}=X_0$, which means that $x_{-j}^{-1}(X_0)\geq 1$. It follows that  $Q_1(X_0)\geq x_{-j}^{-1}(X_0)\geq 1\geq\theta_j$, and therefore $Q_1^{-1}(\theta_j)\geq X_0$. This contradicts with that $Q_1^{-1}(\theta_j)=X_1<X_0$.

If $x_{-j}^{-1}$ is continuous at $X_0$, we know that $x_{-j}^{-1}(X_0)=\theta_j-x_j^{-1}(X_0)=\theta_j-\phi_0$. We also have $\tilde{x}_{j}^{-1}(X_0)\leq \phi_0$ because $\tilde{x}_{j}(\phi_0)\leq X_0$. Denote $q_1=Q_1(X_0)$, then $Q_1(X_0)=\tilde{x}_{j}^{-1}(X_0)+x_{-j}^{-1}(X_0)\leq \theta_j$. By \Cref{theorem:contestant-equilibrium-wine22} we know that 

$$q_1-\Phi_j^*(q_1;\tilde{x}_j(\phi),x_{-j}^{-1}(x))=x_{-j}^{-1}(X_0)=\theta_j-\phi_0.$$ 
By the definition of cummulative choice strategy, $q-\Phi_j^*(q;\tilde{x}_j(\phi),x_{-j}^{-1}(x))$ is non-decreasing in $q$, so it follows that 
\begin{align*}
    &\theta_j-\phi_1=\theta_j-\Phi_j^*(\theta_j;\tilde{x}_j(\phi),x_{-j}^{-1}(x))\\
    \geq &q_1-\Phi_j^*(q_1;\tilde{x}_j(\phi),x_{-j}^{-1}(x))=\theta_j-\phi_0,
\end{align*}
i.e., $\phi_1\leq \phi_0$, which contradicts.

For the second statement, assuming $x_{j}(\phi_0)\leq\tilde{x}_{j}(\phi_0)$, we prove $\phi_1\geq\phi_0$ by contradiction.
Suppose $\phi_1<\phi_0$, then $\tilde{x}_{j}(\phi_1)>\tilde{x}_{j}(\phi_0)\geq x_j(\phi_0)$. Denote $X_0=x_j(\phi_0)$, $X_1=\tilde{x}_j(\phi_1)$, and we know $X_1>X_0$.

From $\phi_0>0$, we have $Q_0^{-1}(\theta_j)=X_0$.
By \Cref{theorem:contestant-equilibrium-wine22} we have $Q_1^{-1}(\theta_j)\geq X_1$, where the equality holds if $\phi_1>0$. Define $X_1'=Q_1^{-1}(\theta_j)$, then it holds that $\tilde{x}_j^{-1}(X_1')=\phi_1$. Therefore, we get $x_{-j}^{-1}(X_1')\geq \theta_j-\phi_1$. Since $X_1'\geq X_1>X_0$, we have 
\begin{align*}
    &\lim_{x\to X_0+0}Q_0(x)=\lim_{x\to X_0+0}(x_j^{-1}(x)+x_{-j}^{-1}(x))\\
    \geq &\phi_0+\theta_j-\phi_1>\theta_j,
\end{align*}
contradicting with that $X_0=Q_0^{-1}(\theta_j)=\sup\{x:Q_0(x)\geq \theta_j\}$.

This completes the proof of the claim.
\end{proof} 
% Since $X_0=\inf\{x:x_j^{-1}(x)+x_{-j}^{-1}(x)<\theta_j\}$, we have $x_j^{-1}(x)+x_{-j}^{-1}(x)<\theta_j$ for all $x>X_0$. By continuity of $x_j^{-1}$ we have $\lim_{x\to X_0+0}x_{-j}^{-1}(X_0)\leq\theta_j-x_j^{-1}(X_0)=\theta_j-\phi_0$. It follows that $x_{-j}^{-1}(X_2)\leq\theta_j-\phi_0$.
% % We discuss on whether $\tilde{x}_j(\phi)$ is constant:

% % (a) If $\tilde{x}_j(\phi)$ is a constant, then $\tilde{x}_j(\phi)=X_2$ for all $\phi\in[0,1]$. It follows that $\theta_j-\Phi_j^*(\theta_j;\tilde{x}_j(\phi),x_{-j}^{-1}(x))\leq x_{-j}^{-1}(X_2)\leq\theta_j-\phi_0$, i.e., $\phi_1\geq \phi_0$, which contradicts.

% Since $\tilde{x}_j(\phi)$ is strictly decreasing, then $\tilde{x}_j^{-1}(X_2)=\phi_0$. Therefore $Q_1(X_2)=\tilde{x}_{j}^{-1}(X_2)+x_{-j}^{-1}(X_2)\leq\theta_j$. Take $q_2=Q_1(X_2)$, then $\Phi_j^*(q_2;\tilde{x}_j(\phi),x_{-j}^{-1}(x))=\phi_0$. It follows that $\phi_1=\Phi_j^*(\theta_j;\tilde{x}_j(\phi),x_{-j}^{-1}(x))\geq \Phi_j^*(q_2;\tilde{x}_j(\phi),x_{-j}^{-1}(x))=\phi_0$, which contradicts.

% 

For any prize structure $\vec{w}_j$, we know that 

$$x_{\vec{w}_j}(\phi)=\sum_{k=1}^nk(w_{j,k}-w_{j,k+1})\xi_k(\phi),$$ 
where $\sum_{k=1}^nk(w_{j,k}-w_{j,k+1})=\sum_{k=1}^nw_{j,k}\leq T_j$. Therefore, for any $\phi_0\in[0,1]$, we obtain that $\max_{\vec{w}_j}x_{\vec{w}_j}(\phi_0)=\max_{k\in[n]}T_j\xi_k(\phi_0)$.

Suppose that $\vec{w}_j^*$ is an optimal prize structure such that $\Phi_j^*(\theta_j;x_{\vec{w}_j^*}(\phi),x_{-j}^{-1}(x))=\phi^*=\frac1n\max_{\vec{w}_j} \hat{R}_j(x_{\vec{w}_{j}}(\phi),$ $x_{-j}^{-1}(x),\theta_j)$. Take arbitrary $k^*\in\arg\max_{k\in[n]}\xi_k(\phi^*)$. We discuss the following three cases:

\noindent \textbf{(a)} $w_{j,1}^*>w_{j,n}^*$ and $k^*<n$. In this case, both $x_{\vec{w}_j^*}(\phi)$ and $x_{\vec{w}_j^{(k^*,T_j)}}$ are strictly decreasing. Therefore, by \Cref{claim.thm4.7}, we have 

$$\hat{R}_j(x_{\vec{w}_j^{(k^*,T_j)}},x_{-j}^{-1}(x),\theta_j)\geq \hat{R}_j(x_{\vec{w}_j^*}(\phi),x_{-j}^{-1}(x),\theta_j),$$
i.e., the simple contest $\vec{w}_j^{(k^*,T_j)}$ is also an optimal prize structure and satisfies the requirement.

\noindent \textbf{(b)} $w_{j,1}^*=w_{j,n}^*$. In this case, $x_{\vec{w}_j^*}(\phi)$ is a constant. We denote $w_0=w_{j,1}^*$. According to \Cref{theorem:contestant-equilibrium-wine22}, define $M_1^+=\{j'\in[m]\setminus \{j\}:w_{j',1}=w_{j',n}=w_0\}$, then we get 

$$\phi^*=\frac1{1+|M_1^+|}(\theta_j-\lim_{x\to w_0+0}x_{-j}^{-1}(x)).$$

We further discuss on $k^*$:

\noindent \textbf{(b.1)} If $k^*<n$, we can construct another optimal prize structure $\vec{w}_j^{**}=\vec{w}_j^{(k^*,w^{**})}$, where $w^{**}=\frac{w_0}{\xi_{k^*}(\phi^*)}$. Then we know that $x_{\vec{w}_j^{**}}(\phi)$ is strictly decreasing in $\phi$, and $x_{\vec{w}_j^{**}}(\phi^*)=w_0$. Moreover, it is not hard to see that 

$$\Phi_j^*(\theta_j;x_{\vec{w}_j^{**}}(\phi),x_{-j}^{-1}(x))=\theta_j-\lim_{x\to w_0+0}x_{-j}^{-1}(x)\geq\phi^*.$$ 
As $\vec{w}_j^*$ is optimal, it must be the case that $\phi^*=\Phi_j^*(\theta_j;x_{\vec{w}_j^{**}}(\phi),x_{-j}^{-1}(x))$. By applying case (a) on $\vec{w}_j^{**}$ and $\vec{w}_j^{(k^*,T_j)}$, we can obtain the desired result.
% Since $\phi^*>0$, we have $x_{-j}^{-1}(x)$ is continuous 

\noindent \textbf{(b.2)} If $k^*=n$, then $x_{\vec{w}_j^{(k^*,T_j)}}(\phi)$ is also a constant function. Specifically, for all $\phi\in[0,1]$, it holds that  $x_{\vec{w}_j^{(k^*,T_j)}}(\phi)=\frac{1}{n}T_j\geq w_0$. If $w_0=\frac{1}{n}T_j$, then $\vec{w}_j^*=\vec{w}_j^{(k^*,T_j)}$ and the statement holds. If $w_0<\frac{1}{n}T_j$, then it is not hard to see that 
\begin{align*}
    &\theta_j-\Phi_j^*(\theta_j;x_{\vec{w}_j^{(k^*,T_j)}}(\phi),x_{-j}^{-1}(x))\\
    =&x_{-j}^{-1}(\frac{1}{n}T_j)
    \leq \lim_{x\to w_0+0}x_{-j}^{-1}(x).
\end{align*}
Consequently, we obtain

$$\Phi_j^*(\theta_j;x_{\vec{w}_j^{(k^*,T_j)}}(\phi),x_{-j}^{-1}(x))\geq\theta_j-\lim_{x\to w_0+0}x_{-j}^{-1}(x)\geq\phi^*.$$ 
Therefore $\vec{w}_j^{(k^*,T_j)}$ is also an optimal prize structure and satisfies the requirement.

\noindent \textbf{(c)} $w_{j,1}^*>w_{j,n}^*$ and $k^*=n$. In this case $x_{\vec{w}_j^{(k^*,T_j)}}(\phi)$ is also a constant function such that for all $\phi\in[0,1]$, i.e., $x_{\vec{w}_j^{(k^*,T_j)}}(\phi)=\frac{1}{n}T_j$. Without loss of generality we can assume that $x_{\vec{w}_j^*}(\phi^*)<x_{\vec{w}_j^{(k^*,T_j)}}(\phi^*)=\frac{1}{n}T_j$, otherwise we can take $n-1\in\arg\max_{k\in[n]}\xi_k(\phi^*)$ and apply case (a). Thus, we have 
\begin{align*}
    \Phi_j^*(\theta_j;x_{\vec{w}_j^{(k^*,T_j)}}(\phi),x_{-j}^{-1}(x))=&\theta_j-x_{-j}^{-1}(\frac{1}{n}T_j)\\
    \geq &\theta_j-\lim_{x\to x_{\vec{w}_j^*}(\phi^*)+0}x_{-j}^{-1}(x)\\
    =&\phi^*.
\end{align*}
Therefore $\vec{w}_j^{(k^*,T_j)}$ is also optimal and satisfies the requirement.
\end{proof}

\subsection{Proof of Theorem \ref{thm:participant objective- spe}}
\begin{proof}
By \Cref{thm: participant objective- simple contest}, we can assume all designers use simple contests that satisfy the conditions in \Cref{thm: participant objective- simple contest}. To construct the equilibrium, for each $j\in[m]$, define $\bar{x}_j(\phi):=T_j\max_{k\in[n]}\xi_k(\phi)$. Intuitively, if designer $j$'s utility is $\phi_j^*$ in equilibrium, then $x_j(\phi_j^*)=\bar{x}_j(\phi_j^*)$.

Similar to the notations used in \Cref{theorem:contestant-equilibrium-wine22}, we define $\bar{x}_j^{-1}(x):=\max\{\phi\in[0,1]:\bar{x}_j(\phi)\geq x\}$, $\bar{Q}(x):=\sum_{j\in[m]}\bar{x}_j^{-1}(x)$, and $\bar{Q}^{-1}(q):=\max\{x:\bar{Q}(x)\geq q\}$. 

Let $X^*=\bar{Q}^{-1}(\theta)$. For each $j\in[n]$, set $\phi^*_j=\bar{x}_j^{-1}(X^*)$ and $k^*_j=\arg\max_{k\in[n]}\xi_k(\phi^*_j)$. If there are multiple choices for $k^*_j$, select the smallest. Let each designer $j$ use the prize structure $\vec{w}_j^*=\vec{w}_j^{k^*_j,T_j}$, we prove that these prize structures form an equilibrium.

Let $w^+=\frac{1}{n}\max_{j\in[m]}T_j$, which represents the highest prize that can be equally offered to $n$ contestants in any contest. If $\sum_{j\in[m]}\lim_{x\to w^++0}\bar{x}_j^{-1}(x)\geq \theta$, then we will have $k^*_j<n$ for all contest $j$. Otherwise we will have some contests with $k^*_j=n$. We discuss in these two cases.

In the former case that $\sum_{j\in[m]}\lim_{x\to w^++0}\bar{x}_j^{-1}(x)\geq \theta$, we know $X^*\geq w^+$ and all $x_{\vec{w}_j^*}(\phi)$ are strictly decreasing. Let $Q(x):=\sum_{j\in[m]}x_{\vec{w}_j^*}^{-1}(x)$, and we have $Q^{-1}(\theta)=X^*$. Moreover, for all $j\in[m]$ we get $x_{\vec{w}_j^*}^{-1}(X)=\phi_j^*$, and therefore it holds that

$$\Phi_j^*(\theta;x_{\vec{w}_1^*}(\phi),\cdots,x_{\vec{w}_m^*}(\phi))=x_{\vec{w}_j^*}^{-1}(Q^{-1}(\theta))=\phi_j^*.$$
By \Cref{claim.thm4.7}, %in the proof of \Cref{thm: participant objective- simple contest}, 
since $k^*_j=\arg\max_{k\in[n]}\xi_k(\phi^*_j)$, we know that $\vec{w}_j^*=\vec{w}_j^{k^*_j,T_j}$ is designer $j$'s best response. Therefore $(\vec{w}_1^*,\cdots,\vec{w}_m^*)$ constitutes an equilibrium.

In the latter case that $\sum_{j\in[m]}\lim_{x\to w^++0}\bar{x}_j^{-1}(x)< \theta$, we first prove a claim that $\xi_{n-1}(\frac12)\geq\xi_{n}(\frac12)$. By definition we have 
\begin{align*}
    \xi_{n-1}(\frac12)=&\frac{1}{n-1}\sum_{k=1}^{n-1}\binom{n-1}{k-1}(\frac{1}{2})^{k-1}(1-\frac{1}{2})^{n-k}\\
    =&\frac{1}{n-1}\sum_{l=0}^{n-2}\binom{n-1}{l}(\frac{1}{2})^{n-1}\\
    =&\frac{1}{n-1}(\frac{1}{2})^{n-1}(2^{n-1}-1)\\
    =&\frac{1-2^{-(n-1)}}{n-1}
\end{align*}
and $\xi_{n}(\frac12)=\frac{1}{n}$. Since $n\geq 2$, we have $1-2^{-(n-1)}\geq 1-\frac1{n-1}=\frac{n-1}{n}$, and consequently $\xi_{n-1}(\frac12)\geq\xi_n(\frac12)$.

Based on this claim, for any $j\in \arg\max_{j\in[m]}$ we have 
\begin{align*}
    \lim_{x\to w^++0}\bar{x}_{j}^{-1}(x)\geq &\lim_{x\to w^++0}x_{\vec{w}_j^{n-1,T_j}}^{-1}(x)\\
    =&x_{\vec{w}_j^{n-1,T_j}}^{-1}(w^+)\geq \frac12.
\end{align*}
This implies that $|\arg\max_{j\in[m]}T_j|=1$, because otherwise $\sum_{j\in[m]}\lim_{x\to w^++0}\bar{x}_j^{-1}(x)\geq 1\geq\theta$.

Let $j^+\in\arg\max_{j\in[m]}T_j$ be the unique designer with maximum budget. We know that $w^+=\frac1n T_{j^+}$, and $X^*=w^+$. Therefore, for designer $j^+$, we have 

$$\Phi_{j^+}^*(\theta;x_{\vec{w}_1^*}(\phi),\cdots,x_{\vec{w}_m^*}(\phi))=\theta-\sum_{j'\neq j}\bar{x}_j^{-1}(w^+).$$
For any other strategy of $j^+$ denoted by $\vec{w}_{j^+}'$, let $Q(x)=\sum_{j\neq j^+}x_{\vec{w}_j^*(\phi)}^{-1}(x)+x_{\vec{w}_{j^+}'}^{-1}(x)$, and we will have $Q^{-1}(\theta)\leq w^+$. Therefore, it holds that
$$\Phi_{j^+}^*(\theta;x_{\vec{w}_{j^+}'}(\phi),x_{\vec{w}_1^*}(\phi),\cdots,x_{\vec{w}_m^*}(\phi))\leq\theta-\sum_{j'\neq j}\bar{x}_j^{-1}(w^+).$$
In other words, $\vec{w}_{j^+}^*$ is the best response for $j^+$. For any other designers $j\neq j^+$, similar to the former case, we can also check that $\vec{w}_{j}^*$ is the best response. Therefore $(\vec{w}_1^*,\cdots,\vec{w}_m^*)$ constitutes an equilibrium.

Finally, since both $\bar{x}_j^{-1}(x)=\sup\{\phi\in[0,1]:\max_{k\in[n]}$ $\xi_k(\phi)\geq \frac{x}{T_j}\}$ and $X^*=\sup\{x:\sum_{j\in[n]}\bar{x}_j^{-1}(x)\leq\theta\}$ can be computed by binary search, this equilibrium can be computed in polynomial time.
    % \qed
\end{proof}

\end{document}